\ifx\documentclass\undefined
\documentstyle[12pt]{article}
\else
\documentclass[12pt]{article}
\fi

\usepackage{color}

\usepackage{ascmac}

\newcommand{\beq}{\begin{eqnarray*}}
\newcommand{\eeq}{\end{eqnarray*}}
\makeatletter
\renewcommand{\theequation}{\thesection.\arabic{equation}}
\@addtoreset{equation}{section}
\def\eqnarray{%
\stepcounter{equation}%
\let\@currentlabel=\theequation
\global\@eqnswtrue
\global\@eqcnt\z@
\tabskip\@centering
\let\\=\@eqncr
$$\halign to \displaywidth\bgroup\@eqnsel\hskip\@centering
$\displaystyle\tabskip\z@{##}$&\global\@eqcnt\@ne
\hfil$\displaystyle{{}##{}}$\hfil
&\global\@eqcnt\tw@$\displaystyle\tabskip\z@{##}$\hfil
\tabskip\@centering&\llap{##}\tabskip\z@\cr}
\makeatother
\newtheorem{theorem}{Theorem}[section]
\newtheorem{lemma}[theorem]{Lemma}

\newtheorem{proposition}[theorem]{Proposition}
\newtheorem{remark}{Remark}[section]

\newsavebox{\toy}
\savebox{\toy}{\framebox[0.65em]{\rule{0cm}{1ex}}}
\newcommand{\QED}{\usebox{\toy}}
\def\nlni{\par\ifvmode\removelastskip\fi\vskip\baselineskip\noindent}
\newenvironment{proof}{\nlni\begingroup\it Proof.\rm}{
\endgroup\vskip\baselineskip}

\begin{document}
\setlength{\baselineskip}{15pt}
\title{
Poisson statistics
for 1d Schr\"odinger operators 
with random decaying potentials
}
\author{
Shinichi Kotani
\thanks{
Department of Mathematics, 
Osaka University, 
Machikaneyamachou 1-1, 
Toyonaka, 
Osaka, 560-0043, Japan.
e-mail : skotani@outlook.com}
\and
Fumihiko Nakano
\thanks{
Department of Mathematics,
Gakushuin University,
1-5-1, Mejiro, Toshima-ku, Tokyo, 171-8588, Japan.
e-mail : 
fumihiko@math.gakushuin.ac.jp}
}
\maketitle
\begin{abstract}
We consider 
the 1d Schr\"odinger operators 
with random decaying potentials 
in the sub-critical case 
where the spectrum is pure point. 
We show that 
the point process composed of the rescaled eigenvalues in the bulk, together with those  zero points of the corresponding eigenfunctions, converges to a Poisson process.
\end{abstract}


\section{Introduction}
The 1d Schr\"odinger operators 
with random decaying potentials 
are known to have rich spectral properties depending on the decay order of the potentials (e.g., \cite{KU, KLS}).
Recently, 
the level statistics problem of this operators are studied and turned out to be related to the 
$\beta$-ensembles which appear in the random matrix theory\cite{KS, KVV, KN, N2}.
In this paper
we consider the following Hamiltonian. 
\beq
H &:=&
-\frac {d^2}{dt^2} + a(t) F(X_t) 
\quad
\mbox{ on } 
L^2 ({\bf R})
\eeq
where the function 
$a \in C^{\infty}({\bf R})$ 
is a decay factor satisfying 
$a(-t) = a(t)$, 
being non-increasing for 
$t > 0$,
and 
\[
a(t) = t^{- \alpha}(1 + o(1)), 
\quad
a'(t) = O(t^{-\alpha-1}), 
\quad
t \to \infty, 
\quad
\alpha > 0. 
\]
The assumption on 
$a'$ 
is technical but we need it to estimate some error terms.
$F(X_t)$ 
is a random factor where 
$F \in C^{\infty}(M)$, 
$M$ 
is the d-dimensional torus, and 
\[
\langle F \rangle
:=
\int_M F(x) dx = 0.
\]
$\{ X_t \}_{t \in {\bf R}}$
is the Brownian motion on 
$M$.
Since 
the potential 
$a(t) F(X_t)$ 
is compact w.r.t. the free Laplacian 
$-d^2 / d t^2$, 
the essential spectrum of 
$H$ 
is equal to 
$\sigma_{ess} (H) = [0, \infty)$ 
which is 
\cite{KU}
(1) 
$\alpha> 1/2$ : 
absolutely continuous, 
(2)
$\alpha < 1/2$ : 
pure point with (sub)exponentially decaying eigenfunctions, 
and 
(3)
$\alpha = 1/2$ : 
there exists a non-random number 
$E_c \ge 0$ 
such that 
the spectrum is pure point on 
$[0, E_c]$ 
and 
singular continuous on 
$[E_c, \infty)$. 

The purpose 
of this paper is to study the local fluctuation of the eigenvalues in the positive energy axis. 
In order for that, 
let 
$H_L := H |_{[0, L]}$ 
be the restriction of 
$H$ 
on the interval 
$[0,L]$
with Dirichlet boundary condition, 
and let 
$\{ E_j (L) \}_{j \ge j_0}$ 
($0 < E_{j_0} (L) < E_{j_0+1} (L) < \cdots$)
be the set of positive eigenvalues of 
$H_L$. 
Take 
the reference energy 
$E_0 > 0$ 
arbitrary, and consider the point process
\beq
\xi_L := 
\sum_{j \ge j_0} 
\delta_{ L ( \sqrt{E_j} - \sqrt{E_0} )}
\eeq
where we take the square root of each eigenvalues which corresponds to the unfolding with respect to the integrated density of states 
$N(E) = \pi^{-1} \sqrt{E}$. 
For a Borel measure 
$\mu$
on 
${\bf R}^d$, 
we denote by 
$Poisson (\mu)$ 
the Poisson process on 
${\bf R}^d$ 
with intensity measure 
$\mu$.
Similarly, 
for a constant 
$c>0$, 
we denote by 
$Poisson (c)$ 
the Poisson distribution with parameter 
$c$. 
The first theorem 
of this paper is 
\begin{theorem}
\label{Poisson}
Let 
$\alpha < 1/2$. 
Then 
$\xi_L$ 
converges in distribution to the Poisson process of intensity 
$d \lambda / \pi$ 
\footnote{
We consider 
the vague topology on the space of point measures on 
${\bf R}$. 
Hence 
$\xi_L \stackrel{d}{\to} \xi$ 
is equivalent to 
$\lim_{L\to \infty}
{\bf E}[ e^{- \xi_L(f)} ] = {\bf E}[ e^{- \xi(f)} ]$
for any 
$f \in C_c^+ ({\bf R})$.
} 
\beq
\xi_L \stackrel{d}{\to} 
\mbox{Poisson}
\left( \frac {d \lambda}{\pi} \right), 
\quad
L \to \infty.
\eeq
\end{theorem}
\begin{remark}
When we consider 
two reference energies 
$E_1$, $E_2$, $E_1 \ne E_2$, 
then the corresponding point processes 
$\xi_1$, $\xi_2$ 
jointly converge to the independent Poisson processes of intensity 
$d \lambda / \pi$.
\end{remark}
\begin{remark}
Together 
with results in 
\cite{KN, N2}, 
we have 
\footnote{
In (2), 
$\beta = \beta(E_0) := 8 E_0 / C(E_0)$ 
where 
$C(E) := \langle \nabla g_{\sqrt{E}}, \nabla g_{\sqrt{E}} \rangle$, 
$g_{\sqrt{E}} := (L + 2i \sqrt{E})^{-1}F$.
$\beta (E)$
is equal to the reciprocal of the Lyapunov exponent of 
$H$.
%
}
\beq
&(1)&\; 
\alpha > \frac 12
\Longrightarrow
\xi_L \to \mbox{clock process}
\\
&(2)& \;
\alpha = \frac 12
\Longrightarrow
\xi_L \to 
\mbox{ Sine$_{\beta}$ process }
\\
&(3)& \;
\alpha < \frac 12
\Longrightarrow
\xi_L \to 
\mbox{Poisson process}
\eeq
Such kind 
of results have been known for discrete models : 
\cite{KS} 
proved (1)-(3) above for CMV matrices, 
\cite{ALS}
proved 
``clock behavior"
(similar to (1)) for Jacobi matrices, 
and 
\cite{KVV} 
proved (2) for 1d discrete Schr\"odinger operators. 
Hence 
our result is a continuum analogue of them. 
The  
model-independent nature of those results is due to the fact that the Pr\"ufer phases of those models obey the similar equations and thus have similar behavior. 
The 
global fluctuation of eigenvalues is studied in \cite{N4} which also shows different behavior in above three cases. 
%
\end{remark}
\begin{remark}
Let 
$H'_L := 
( - \frac {d^2}{dt^2} + L^{- \alpha} F(X_t) )
|_{[0,L]}$
be the Hamiltonian with decaying coupling constant under the Dirichlet boundary condition.
The method of proof of 
Theorem \ref{Poisson} also works for $H'_L$ 
so that 
together with results in 
\cite{N2} 
we have
\footnote{
In (2), 
$\tau = \tau(E_0) = C(E_0) / (2E_0)= 4 / \beta(E_0) $
\cite{N3}.
}
\beq
&(1)&\; 
\alpha > \frac 12
\Longrightarrow
\xi_L \to \mbox{clock process}
\\
&(2)& \;
\alpha = \frac 12
\Longrightarrow
\xi_L \to 
\mbox{
Sch$_{\tau}$ 
process
}
\\
&(3)& \;
\alpha < \frac12
\Longrightarrow
\xi_L
\to 
\mbox{Poisson}
\left( \frac {d \lambda}{\pi} \right)
\eeq
\cite{KVV} 
proved (2) for 1d discrete Schr\"odinger operators. 
\end{remark}
\begin{remark}
It would be interesting 
to study the behavior of eigenvalues near the bottom edge of the essential spectrum
(i.e., 
to study 
$\xi_L$ 
for 
$E_0 = 0$), 
for which the technique in this paper does not apply.
For 
recent development in this respect, we refer to \cite{AD2}. 
\end{remark}
To see the outline of proof, 
we introduce the Pr\"ufer variable as follows.
Let 
$x_t$ 
be the solution to the Schr\"odinger equation 
$H x_t = \kappa^2 x_t$, $x_0 = 0$, 
which is represented in the following form.
\beq
\left( \begin{array}{c}
x_t \\ x'_t / \kappa
\end{array} \right)
=
r_t(\kappa) 
\left( \begin{array}{c}
\sin \theta_t(\kappa) 
\\
\cos \theta_t(\kappa)
\end{array} \right), 
\quad
\theta_0(\kappa) = 0.
\eeq
Set 
\beq
\Theta_L(\lambda)
&:=&
\theta_L
\left(
\sqrt{E_0} + \frac {\lambda}{L}
\right)
-
\theta_L
\left(
\sqrt{E_0}
\right)
\\
\phi_L(E_0)
&:=&
\{ 
\theta_L (\sqrt{E_0})
\}_{\pi}, 
\quad
\mbox{ where }
\quad
\{ x \}_{\pi}
:=
x  - 
\left\lfloor
\frac {x}{\pi}
\right\rfloor
\pi.
\eeq
Since, 
by Sturm's oscillation theorem, 
$E = E_j(L)$
if and only if 
$\theta_L (\sqrt{E}) = j \pi$, 
the Laplace transform of 
$\xi_L$ 
has the following representation.
\begin{eqnarray}
{\bf E}\left[
e^{- \xi_L(f)}
\right]
&=&
{\bf E}\left[
\exp 
\left(
- \sum_k 
f 
\left(
\Theta_L^{-1}
(k \pi - \phi_L(E_0))
\right)
\right)
\right]
\label{Laplace}
\\
\mbox{where}
\quad
\xi_L(f)
&=&
\int_{\bf R} f(x) \xi_L(dx), 
\quad
f \in C_c^+({\bf R}).
\nonumber
\end{eqnarray}
Thus 
our aim is to study the joint limit of 
$( \Theta_L(\lambda), \phi_L (E_0) )$.
Here 
we replace 
$L$ 
by 
$n$ 
and consider the family 
$H_{nt}$ ($t \in [0,1]$)
of Hamiltonians. 
We will 
show that the following limits exist. 
\beq
\widehat{\Theta}_t(\lambda)
&\stackrel{d}{=}&
\lim_{n \to \infty} \Theta_{nt}(\lambda), 
\quad
\widehat{\phi}_t
\stackrel{d}{=}
\lim_{n \to \infty}
\phi_{nt}(E_0).
\eeq
In the first equation, 
both sides are regarded as the non-decreasing function(with the weak topology as a measure)-valued processes in $t$.
Then 
we have the following theorem. 
\begin{theorem}
\label{main}
\mbox{}\\
(1)
For any 
$t \in (0,1]$, 
$\widehat{\phi}_t$ 
is uniformly distributed on 
$[0, \pi)$. \\
(2)
\[
\widehat{\Theta}_t(\lambda)
=
\pi
\int_{[0,t] \times [0, \lambda]}
\widehat{P}(ds d\lambda')
\]
where 
$\widehat{P}
=
Poisson
\left(
\pi^{-1}
1_{[0,1]}(s)
ds d \lambda'
\right)$
is the Poisson process on 
${\bf R}^2$ 
whose intensity measure is equal to 
$\pi^{-1}
1_{[0,1]}(s)
ds d \lambda'$.
\end{theorem}
\begin{remark}
The statement 
in Theorem \ref{main}(2) 
is conjectured in \cite{KS} for CMV matrices.
On the other hand, 
for the Anderson model 
$H = - \triangle + V_{\omega}(x)$ 
on 
$l^2 ({\bf Z}^d)$,
the following facts are known \cite{KiN, N1}.
Let 
$H_L := H |_{ \{ 1, \cdots, L \}^d }$ 
be the restriction of 
$H$ 
on the box of size 
$L$, 
with 
$\{ E_j (L) \}_{j \ge 1}$ 
being its eigenvalues. 
Let 
$x_j (L) \in {\bf R}^d$ 
be the localization center corresponding to  
$E_j(L)$.
If 
$E_0$ 
lies in the localized region, 
we have 
\begin{equation}
\sum_{j}
\delta_{
\left(
L^d (E_j(L) - E_0), \, L^{-1} x_j (L)
\right)
}
\stackrel{d}{\to}
Poisson 
\left(
n(E_0) 1_{[0,1]^d}(x) dE \times dx
\right)
\label{LocCtr}
\end{equation}
where 
$n(E_0):= 
\frac {d}{dE}N(E)|_{E=E_0}$ 
is the density of states at 
$E=E_0$. 
%

%
The jump points 
of the function 
$t \mapsto 
\left\lfloor
\Theta_{nt}(\lambda)/\pi
\right\rfloor$
are (modulo some errors) related to the zero points of the eigenfunction such that the corresponding eigenvalue is less than 
$\lambda$.
Since 
the eigenfunction decays sub-exponentially 
and since 
the set of jump points of the function 
$t \mapsto 
\hat{\Theta}_{t}(\lambda)/\pi
$ 
has the monotonicity in 
$\lambda$ 
to be described in eq.(\ref{independence}), 
those jump points are close to the localization center of each eigenfunctions. 
Hence 
we believe that the statement like eq.(\ref{LocCtr}) 
holds also for our case and that Theorem \ref{main} (2) is related to this speculation.
\end{remark}
We shall 
explain the idea of proof.
The Pfr\"ufer phase 
satisfies the integral equation 
(\ref{integralequation}) 
by which we compute the equation satisfied by 
$\Theta_{nt}(\lambda)$.
By 
using ``Ito's formula"
(\ref{Itoformula}) 
we can show that, up to error terms, 
\beq
d \Theta_{nt}(\lambda)
&\sim&
\lambda dt
+
n^{\frac 12 - \alpha}
Re \left[
\left(
e^{2i \Theta_{nt}(\lambda)}-1
\right)
t^{- \alpha}
d Z_t
\right]
\eeq
where 
$Z_t = X_t + i Y_t$
is the complex Brownian motion. 
At this point, 
we have a general picture : 
(1) 
$\alpha > 1/2$ : 
second term vanishes which implies the convergence to the clock process, 
(2)
$\alpha = 1/2$ : 
$\Theta_{nt}(\lambda)$ 
converges to the solution to a SDE, and 
(3)
$\alpha < 1/2$ : 
the diffusion term will be dominant so that 
$\Theta_{n t}(\lambda)$ 
should be in a vicinity of 
$\pi {\bf Z}$ 
in order to have 
$(e^{2 i \Theta_{n t}(\lambda)} - 1)$ 
small. 
Here we note that 
$\Theta_{nt}(\lambda) > 0$ 
for 
$\lambda > 0$ 
and 
${\bf E}[ \Theta_{nt}(\lambda) ] 
= 
\lambda t + o(1)$
(Proposition \ref{Thetalimit}).
By the change of variables 
\[
t = s^{\gamma}, 
\quad
\gamma := \frac {1}{1-2\alpha}, 
\quad
s \in [0,1], 
\]
we have 
\beq
d \Theta_{n s^{\gamma}} (\lambda)
\sim
\lambda \gamma s^{\gamma-1} d s 
+
n^{\frac {1}{2 \gamma}} 
Re \left[
(e^{2i \Theta_{n s^{\gamma}}(\lambda)}-1)
d \tilde{Z}_s
\right].
\eeq
Here 
we recall the definition of the 
Sine$_{\beta}$-process \cite{VV}.
Let 
$\alpha_t(\lambda)$ 
be the solution to the following SDE. 
\begin{eqnarray}
d \alpha_t (\lambda)
&=&
\lambda \cdot \frac {\beta}{4} 
e^{- \frac {\beta}{4} t}
dt
+
Re \left[
\left(
e^{i \alpha_t(\lambda)}-1
\right)
d Z_t
\right]
\label{Sine}
\\
\alpha_0(\lambda) &=& 0.
\nonumber
\end{eqnarray}
Then 
the function 
$t \mapsto 
\left\lfloor
\alpha_t(\lambda) / 2 \pi
\right\rfloor$
is non-decreasing and the limit 
$\alpha_{\infty}(\lambda)
:=
\lim_{t \to \infty} \alpha_t(\lambda)$
satisfies 
$\alpha_{\infty}(\lambda) \in 2 \pi {\bf Z}$, a.s.
Then 
Sine$_{\beta}$-process  
on the interval
$[\lambda_1, \lambda_2]$
is defined by 
\[
Sine_{\beta} [ \lambda_1,\lambda_2 ]
\stackrel{d}{=}
\frac {
\alpha_{\infty}(\lambda_2)- \alpha_{\infty}(\lambda_1)
}
{ 2 \pi }.
\]
Allez-Dumaz \cite{AD} 
showed that 
Sine$_{\beta} 
\stackrel{d}{\to} 
Poisson (d \lambda / 2 \pi)$
as 
$\beta \to 0$. 
This fact 
can easily be generalized to other processes where the drift term in the corresponding SDE
(\ref{Sine}) 
is replaced by functions 
$f$
with mild conditions\cite{N3}.
Moreover, 
by a scaling 
$t \mapsto \frac {\beta}{4} t$, 
eq.(\ref{Sine}) 
becomes 
\begin{eqnarray*}
d \alpha_t (\lambda)
&=&
\lambda  
e^{- t}dt
+
\frac {2}{\sqrt{\beta}}
Re \left[
\left(
e^{i \alpha_t(\lambda)}-1
\right)
d Z_t
\right]
\\
\alpha_0(\lambda) &=& 0
\end{eqnarray*}
so that, by setting 
$\beta = n^{-\frac {1}{\gamma}}$,  
we 
can use the idea of \cite{AD} :
to study 
the hitting time of 
$\Theta_{nt}(\lambda)$ 
to the set 
$\pi {\bf Z}$, 
we consider 
\beq
R(nt)
:=
\log \tan \frac {\Theta_{nt}(\lambda)}{2}, 
\eeq
SDE of which has a diffusion term with constant coefficient so that we may use comparison argument. 
In fact, 
modulo error terms, 
we have (Propositions \ref{Ricatti}, \ref{comparison})
\begin{eqnarray}
d R(n t^{\gamma})
&\sim&
\left(
\lambda \gamma t^{\gamma-1}
\cosh R(n t^{\gamma})
+
\frac {C_n^2}{2} \tanh R(n t^{\gamma})
\right) dt
+
C_n dM_t 
\quad\quad
\label{Ricattisim}
\\
&\mbox{where}&\quad
C_n = C(E_0, F) n^{\frac {1}{2 \gamma}}, 
\quad
d\langle M \rangle_t
=
(1 + o(1)) \, dt, 
\nonumber
\end{eqnarray}
and
$C(E_0, F)$ 
is a positive constant depending on 
$E_0$, $F$. 
Here 
we use assumptions on 
$a$, $a'$ 
to estimate error terms. 
By a time-change, 
we can suppose that 
$M_t$ 
is a Brownian motion.
We divide the interval 
$[0,1]$ 
into small random ones  
$I_k = 
\left[
\tau_k/N, \tau_{k+1}/N 
\right]$
and consider the stationary processes 
$S_{\pm}$ 
which are the solution to the following SDE's on each 
$I_k$'s. 
\beq
d S_{+}(t)
&\sim&
\left(
\lambda \gamma
\left( \frac {\tau_{k+1}}{N} \right)^{\gamma-1}
\cosh S_{+}(t)
+
\frac {C_n^2}{2} \tanh S_{+}(t)
\right)dt
+
C_n d M_t
\\
d S_{-}(t)
&\sim&
\left(
\lambda \gamma
\left( \frac {\tau_k}{N} \right)^{\gamma-1}
\cosh S_{-}(t)
+
\frac {C_n^2}{2} \tanh S_{-}(t)
\right)dt
+
C_n d M_t.
\eeq
On each 
$I_k$, 
we can bound 
$R(n t^{\gamma})$ 
by 
$S_{\pm}$ 
from above and below : 
\beq
S_{-}(t) \le R(n t^{\gamma}) \le S_{+}(t).
\eeq
We can 
explicitly compute the explosion times of 
$S_{\pm}$ 
which converge to 
$Exp \left(\tilde{\lambda}/\pi\right)$
as 
$n \to \infty$,  
where 
$\tilde{\lambda}
:=
\lambda \gamma
\left( \tau_{k+1}/N \right)^{\gamma-1}$
(Proposition \ref{exponential}).
By an argument 
like the convergence of Riemannian sums to the integral, we can show that 
the jump points of the function
$s \mapsto \left \lfloor
\Theta_{n s^{\gamma}}(\lambda) / \pi 
\right \rfloor$
converge to 
$Poisson \left(
\pi^{-1}
\gamma s^{\gamma-1} 1_{[0,1]}(s) ds 
\right)$(Proposition \ref{marginal}).
Hence 
for an interval 
$J \subset {\bf R}$, 
$\xi_L(J)$ 
converges to the Poisson distribution with parameter 
$\pi^{-1} |J|$. 
It then 
suffies to show that the collection of random variables 
$\xi_L(J_1), \cdots, \xi_L(J_n)$
converge jointly to the independent ones for disjoint intervals
$J_1, J_2, \cdots, J_n$. 
For 
$\lambda_1 < \lambda_2$, 
let 
$P_{\lambda_1}$, 
$P_{\lambda_2}$, 
$P_{\lambda_1, \lambda_2}$ 
be the limit of those point processes composed by the jump points of functions 
$s \mapsto \left \lfloor
\Theta_{n s^{\gamma}}(\lambda_1) / \pi  
\right \rfloor$, 
$\left \lfloor
\Theta_{n s^{\gamma}}(\lambda_2) / \pi  
\right \rfloor$
and 
$\left \lfloor
(
\Theta_{n s^{\gamma}}(\lambda_2)
-
\Theta_{n s^{\gamma}}(\lambda_1)
)/
\pi
\right \rfloor$ 
respectively.
Then 
$P_{\lambda_1}$, 
$P_{\lambda_2}$, 
$P_{\lambda_1, \lambda_2}$
turn out to be the 
${\cal F}_s$-Poisson processes
under a suitable choice of the filtration 
${\cal F}_s$
(Lemma \ref{jointlyPoisson}). 
Letting 
${\cal P}_{\lambda_1}$, 
${\cal P}_{\lambda_2}$, 
${\cal P}_{\lambda_1, \lambda_2}$ 
be the set of atoms, we show 
(Lemmas \ref{monotonicity2}, \ref{independent})
\begin{equation}
{\cal P}_{\lambda_1}
\subset 
{\cal P}_{\lambda_2}, 
\quad
P_{\lambda_1} 
\cap 
{\cal P}_{\lambda_1, \lambda_2}
=
\emptyset
\label{independence}
\end{equation}
from which the independence of 
${\cal P}_{\lambda_1}$
and 
${\cal P}_{\lambda_1, \lambda_2}$
follows.
Finally 
we show that 
$\lim_{n \to \infty}
\Theta_{nt}(\lambda) / \pi
\in {\bf Z}$,  
a.s. which proves Theorem \ref{main}(2).
The statement in 
Theorem \ref{main}(1) is essentially proved in our previous paper \cite{KN} where the condition 
$\langle F \rangle = 0$ 
is used.
Theorem \ref{Poisson} 
follows from 
eq.(\ref{Laplace}) and 
Theorem \ref{main}.

The rest 
of this paper is organized as follows.
In Section 2, 
we study the behavior of 
$\Theta_{nt} (\lambda)$ 
and derive some properties of 
the expectation of 
$\Theta_{nt}(\lambda)$ 
and the monotonicity of the function 
$t \mapsto \left\lfloor 
\Theta_{nt}(\lambda) / \pi 
\right\rfloor$.
In Section 3, 
we derive the Ricatti equation
(\ref{Ricattisim}) 
satisfied by 
$R(nt)$. 
In Section 4, 
we estimate 
$R(nt^{\gamma})$ 
from above and below by solutions 
$R_{\pm}$ 
to simple SDE's.
In Section 5, 
following the argument in 
\cite{AD}, 
we consider the stationary approximation 
$S_{\pm}$ 
of 
$R_{\pm}$ 
and compute the explosion time of them. 
Then we show that the jump points of the function 
$t \mapsto \left\lfloor 
\Theta_{nt} / \pi 
\right\rfloor$
converge to a Poisson process and that the processes 
$P_{\lambda_1}$ 
and 
$P_{\lambda_1, \lambda_2}$
mentioned above are independent.
In Section 6, 
we prove Theorems 
\ref{Poisson}, \ref{main}.
Sections 7, 8 are appendices.
In what follows, 
$C$, $C'$ 
are positive constants which may change from line to line in each argument.
%
\section{
Behavior of 
$\Theta_{nt}(\lambda)$
}
In this section 
we introduce notations and derive some basic properties of the relative Pr\"ufer phase 
$\Theta_{nt}(\lambda)$.
Let  
$\tilde{\theta}_t(\kappa)$ 
be defined by 
\beq
\theta_t(\kappa)
&=&
\kappa t + \tilde{\theta}_t(\kappa)
\eeq
which satisfies the following integral equation. 
\begin{equation}
\tilde{\theta}_t(\kappa)
=
\frac {1}{2 \kappa}
Re 
\int_0^t 
\left(
e^{2i \theta_s(\kappa)} - 1 
\right)
a(s) F(X_s) ds.
\label{integralequation}
\end{equation}
Set 
\beq
\kappa_0 &:=& \sqrt{E_0}
\\
\kappa_c &:=& \kappa_0 + \frac cn, 
\quad
n > 0, 
\quad
c \in {\bf R}
\\
r_t^{(n)}(m)
&:=&
e^{2mi \theta_t(\kappa_c)}
-
e^{2mi \theta_t(\kappa_0)}, 
\quad
m \in {\bf Z}
\\
A_n(t)
&:=&
- \frac {c}{2\kappa_c \cdot \kappa_0}
Re 
\left(
e^{2i \theta_t(\kappa_c)}-1
\right)
F(X_t)
\\
(\triangle f)(m)
&:=&
\frac 12
\left(
f(m+1) + f(m-1) 
\right) - f(m). 
\eeq
By 
(\ref{integralequation}) 
we have 
\begin{eqnarray}
\Theta_{nt}(c)
&=&
\theta_{nt}(\kappa_c) - \theta_{nt}(\kappa_0)
\nonumber
\\
&=&
ct 
+
\frac {1}{2 \kappa_0} Re 
\int_0^{nt}
r_s^{(n)}(1) a(s) F(X_s) ds
+
\frac 1n
\int_0^{nt} A_n(s) a(s) ds.
\quad
\label{firstcomputation}
\end{eqnarray}
%
\begin{remark}
\label{monotonicity1}
For large 
$n$, 
we can find 
$t_0 > 0$ 
such that for 
$t \ge t_0$, 
we have 
$c > A_n (nt) a(nt)$. 
Then by 
eq.(\ref{firstcomputation}), 
for 
$t \ge t_0$, 
once 
$\Theta_t^{(n)}(\lambda)$
enters to an interval 
$( (k+1) \pi, (k+2) \pi )$ 
for some 
$k \in {\bf N}$, 
it never returns to 
$( k \pi, (k+1) \pi)$. 
In other words, the function 
$t \mapsto \left\lfloor 
\Theta_{nt}(\lambda) / \pi
\right\rfloor$
is non-decreasing.
\end{remark}
Here we make use of the following identity which is a consequence of Ito's formula 
\cite{KU} : 
for 
$f \in C^{\infty}(M)$ 
and 
$\kappa \ne 0$, 
\begin{eqnarray}
e^{i \kappa s} f(X_s) ds
&=&
d \left(
e^{i \kappa s} (R_{\kappa}f)(X_s) 
\right)
-
e^{i \kappa s} (\nabla R_{\kappa}f)(X_s) d X_s
\label{Itoformula}
\\
f(X_s) ds
&=&
\langle f \rangle ds
+
d ((R_0 f)(X_s))
-
\nabla (R_0 f)(X_s) d X_s
\label{SecondItoformula}
\\
\mbox{where}
\quad
R_{\kappa}f
&:=&
(L + i \kappa)^{-1} f, 
\quad
R_0 f := 
L^{-1} (f - \langle f \rangle).
\nonumber
\end{eqnarray}
$L$ 
is the generator of 
$X_t$.
Eq.(\ref{Itoformula}) 
and the integration by parts yields the following equation.
\begin{lemma}
\label{partialintegration}
Let 
$b \in C^{\infty}([0, \infty))$, 
$\varphi \in C^{\infty}(M)$, 
and let 
$g_{\varphi}^{m \kappa_0} 
:=
R_{2m \kappa_0} \varphi
=
(L + 2mi \kappa_0)^{-1} \varphi$. 
Then we have 
\beq
&&
\int_0^t 
b(s) r_s^{(n)}(m) \varphi(X_s) ds
\\
&=&
(-2mi)
\cdot
\frac {1}{2 \kappa_0}
\int_0^t 
b(s)
(\triangle r_s^{(n)})(m)
a(s)F(X_s)g_{\varphi}^{m \kappa_0} (X_s) ds
\\
&&
+
\left[
b(s) r^{(n)}_{s}(m) g_{\varphi}^{m \kappa_0}(X_s)
\right]_0^t
\\
&& - 
\int_0^t b'(s) 
r^{(n)}_{s}(m)
g_{\varphi}^{m \kappa_0}(X_s) ds
\\
&& - 2mi 
\cdot 
\frac 1n
\int_0^t 
b(s)
\left(
c + 
A_n (s) a(s) 
\right)
e^{2mi \theta_s(\kappa_c)}
g_{\varphi}^{m \kappa_0} (X_s) ds
\\
&& - 
\int_0^t 
b(s) 
r^{(n)}_{s}(m)
\nabla g_{\varphi}^{m \kappa_0} (X_s) d X_s.
\eeq
\end{lemma}
Putting 
$m=1$, $\varphi = F$, 
and 
$b(t) = a(t)$ 
in Lemma \ref{partialintegration}, 
we have 
\begin{lemma}
\label{Theta}
\beq
\Theta_{nt} (c)
&=&
ct + M_t^{(n)} + O_t^{(n)} + \delta^{(n)}_t
\eeq
where 
\beq
M_t^{(n)}
&=&
- \frac {1}{2 \kappa_0}
Re
\int_0^{nt} 
a(s) 
r^{(n)}_{s}(1)
\nabla g_{F}^{\kappa_0} (X_s) d X_s
\\
O_t^{(n)}
&=&
\frac {1}{2 \kappa_0}
Re
\left( 
- \frac {2i}{2 \kappa_0}
\int_0^{nt}
a(s)^2 
( \triangle r_s^{(n)} )(1)
F(X_s)
g_F^{\kappa_0}(X_s) ds
\right)
\\
\delta_t^{(n)}
&=&
\frac {1}{2 \kappa_0}
Re
\Biggr\{
\left[
a(s) r_s^{(n)}(1) g_F^{\kappa_0}(X_s) 
\right]_0^{nt}
-
\int_0^{nt}
a'(s) r_s^{(n)}(1) g_F^{\kappa_0}(X_s) ds
\\
&& \qquad
+
(-2i) \frac 1n
\int_0^{nt}
a(s)
\left(
c + A_n(s) a(s)
\right)
e^{2i \theta_s(\kappa_c)}
g_F^{\kappa_0} (X_s) ds
\Biggr\}
\\
&& \qquad + 
\frac 1n
\int_0^{nt} a(s) A_n (s) ds.
\eeq
Moreover, 
\[
\lim_{n \to \infty} \delta_t^{(n)} = 0.
\]
\end{lemma}
By using 
Lemmas \ref{partialintegration}, \ref{Theta} 
we can prove the following Proposition which is necessary to study the behavior of 
${\bf E}[ \Theta_{nt}(\lambda) ]$. 
\begin{proposition}
\label{Thetalimit}
Suppose that 
\[
\int_0^{\infty} a(s)^{j_0} ds < \infty
\]
for some 
$j_0 \ge 1$.
Then for 
$t > 0$, 
we have 
\beq
\Theta_{nt} (c)
&=&
ct + \widetilde{M}_t^{(n)} + o(1), 
\quad
n \to \infty
\eeq
where 
$\widetilde{M}_t^{(n)}$ 
is a martingale.
\end{proposition}
\begin{proof}
Note that 
$\lim_{n \to \infty}r_s^{(n)}(m) = 0$.
If 
$j_0 \le 2$, 
$O_t^{(n)} = o(1)$ 
which already proves the statement of Proposition \ref{Thetalimit} 
with 
$\widetilde{M}_t^{(n)} = M_t^{(n)}$.
If 
$j_0 \ge 3$, 
we apply 
Lemma \ref{partialintegration}
for 
$O_t^{(n)}$ 
so that
\beq
O_t^{(n)}
&=&
\frac {1}{2 \kappa}
Re
\Biggl(
- \frac {2i}{2 \kappa}
\int_0^{nt}
a(s)^2
\triangle r_s^{(n)} (1)
F(X_s)
g_F^{\kappa} (X_s) ds
\Biggr)
\\
&=&
Re
\sum_{m=1,2}
C_m
\int_0^{nt}
a(s)^3
r_s^{(n)}(m)
G_m^{(n)}(X_s) ds
+
(martingale)
+
o(1)
\\
\eeq
where 
$G_m^{(n)}$ 
is uniformly bounded.
Iterating 
this process until we have 
$a(s)^{j_0}$ 
yields 
\beq
O_t^{(n)}
&=&
\sum_m c_m 
\int_0^{nt} 
a (s)^{j_0} r_s^{(n)} (m) G_m^{(n)} (X_s) ds
+
(martingale)
+
o(1).
\eeq
\QED
\end{proof}
%
\section{Ricatti equation}
For a function 
$\kappa \mapsto f(\kappa)$ 
we introduce 
\[
\Delta
f := f(\kappa_c) - f(\kappa_d), 
\quad
0 \le d < c, 
\quad
\kappa_x := 
\kappa_0 + \frac xn.
\]
This 
definition is different from that in Section 2.
To study 
the hitting time of 
$\Theta_{nt}(\lambda)$ 
to the set 
$\pi {\bf Z}$, 
or that of  
$(\Theta_{nt}(\lambda') - \Theta_{nt}(\lambda))$ 
in general, we consider
\beq
R(t)
&:=&
\log \tan\frac { \Delta \theta_t}{2}.
\quad
\eeq
Note that 
\begin{equation}
\cosh R(s) = 
\frac {1}{\sin \Delta \theta_s}, 
\quad
\sinh R(s) = 
- \frac {\cos \Delta \theta_s}{\sin \Delta \theta_s}.
\label{logtan}
\end{equation}
Here 
we recall that, for 
Sine$_{\beta}$-process, the corresponding process 
$\tilde{R}(t) := \log \tan 
\left(
\alpha_t(\lambda) / 4
\right)$  
with  
$\alpha_t(\lambda)$ 
being the solution to 
eq.(\ref{Sine}) satisfies
\begin{equation}
d \tilde{R}(t)
=
\frac 12
\left(
\lambda \frac {\beta}{4}
e^{- \frac {\beta}{4} t}
\cosh \tilde{R}(t) 
+
\tanh \tilde{R}(t)
\right)dt
+
d B_t.
\label{SinebetaRicatti}
\end{equation}
%
%
The following 
Proposition implies that 
$R(nt)$ 
is close to the solution to a SDE which is similar to eq.(\ref{SinebetaRicatti}). 
\begin{proposition}
\label{Ricatti}
\begin{eqnarray}
R(nt) - R(0)
&=&
\frac {c-d}{n}
\int_0^{nt}
\cosh R(s) \, ds
\nonumber\\
+
\frac {1}{2 \kappa_0} 
&Re&
\left[
- \frac {\langle F g_{\kappa_0} \rangle}{\kappa_0}
\right]
\int_0^{nt} 
a(s)^2
\tanh R(s) \,  ds 
+ M_t
+
E(nt)
\qquad
\label{Ricattiequation}
\end{eqnarray}
where 
$M$
is a martingale with 
\begin{eqnarray}
d\langle M \rangle_t
&=&
\left( \frac {1}{2 \kappa_0} \right)^2
2
\langle \psi_{\kappa_0} \rangle
n a(nt)^2 
(1+o(1))
dt, 
\quad
n \to \infty
\label{M}
\\
\psi_{\kappa_0}
&:=&
[g_{\kappa_0}, 
\overline{
g_{\kappa_0}
}
],
\quad
g_{\kappa_0}
:=
R_{2 \kappa_0}F
=
(L + 2i \kappa_0)^{-1} F, 
\quad
[f,g]
:=
\nabla f \cdot \nabla g.
\nonumber
\end{eqnarray}
The last term 
$E(nt)$ 
in eq.(\ref{Ricattiequation}) is an negligible error  compared to 1st and 2nd terms of RHS in eq.(\ref{Ricattiequation}), and 
has the following form. 
\beq
E(nt)
&=&
\int_0^{nt}
\cosh(R(s)) b(s) c_1(s) ds
\\
&& \qquad
+
\int_0^{nt} 
\tanh (R(s)) a(s)^3 c_2(s) ds
+
e^{(n)}(t)
+
C
\eeq
where 
$C$ 
is a non-random constant and 
\beq
&&
b(s) = 
\frac 1n a(s) + a'(s) + a(s)^{j_0}, 
\quad
j_0 := \min
\{ j \in {\bf N} \, | \, 1 - j \alpha < 0 \}
\\
&&
c_1(s), c_2(s) : \mbox{ bounded functions }
\\
&&
e^{(n)}(t) \le C' n^{- \alpha}.
%
\eeq
%
%
%
\end{proposition}
\begin{proof}
First of all, 
we introduce a notation 
$A \approx B$ 
meaning that 
$A - B$ 
is a sum of an negligible error 
$E(nt)$ 
and a martingale 
$N$ 
whose quadratic variation is negligible compared to that of 
$M$
in eq.(\ref{M}) : 
\beq
A \approx B
&\stackrel{def}{\Longleftrightarrow}&
A-B = E(nt) + N_t, 
\quad
d \langle N \rangle_t
\le
C \cdot n a(nt)^3 dt.
\eeq
By the integral equation (\ref{integralequation}), 
we have
\beq
&&
R(nt) - R(0)
\\
&=&
\int_0^{nt}
\frac {1}{\sin (\Delta \theta_s)} 
\frac {d}{ds} 
\left(
\theta_s(\kappa_c) - \theta_s(\kappa_d)
\right)
ds
\\
&=&
\int_0^{nt} 
\frac {1}{\sin (\Delta \theta_s)} 
\Biggl\{
\frac {c-d}{n} 
+
\frac {1}{2\kappa_c}
Re 
\left(
e^{2i \theta_s(\kappa_c)} - 1
\right)
a(s) F(X_s)
\\
&& \qquad
-
\frac {1}{2\kappa_d}
Re 
\left(
e^{2i \theta_s(\kappa_d)} - 1
\right)
a(s) F(X_s)
\Biggr\}
ds
\\
&=&
\int_0^{nt} 
\frac {1}{\sin (\Delta \theta_s)} 
\frac {c-d}{n} 
ds
\\
&& 
+
\int_0^{nt} 
\frac {1}{\sin (\Delta \theta_s)}
\frac {1}{2 \kappa_0}
Re 
\left(
e^{2i \theta_s(\kappa_c)}
-
e^{2i \theta_s(\kappa_d)}
\right)
a(s) F(X_s) ds
\\
&& +
\left(
\frac {1}{2 \kappa_c}- \frac {1}{2 \kappa_0} 
\right)
\int_0^{nt}
\frac {1}{\sin (\Delta \theta_s)} 
Re \left(
e^{2i \theta_s(\kappa_c)} -1
\right)
a(s) F(X_s) ds
\\
&& 
-
\left(
\frac {1}{2 \kappa_d}- \frac {1}{2 \kappa_0} 
\right)
\int_0^{nt}
\frac {1}{\sin (\Delta \theta_s)} 
Re \left(
e^{2i \theta_s(\kappa_d)} -1
\right)
a(s) F(X_s) ds
\\
&=:&
I + \cdots + IV.
\eeq
By (\ref{logtan}), 
$I$ is equal to the 1st term of RHS in 
eq.(\ref{Ricattiequation}).
Since 
$\kappa_c^{-1} - \kappa_0^{-1} = O(n^{-1})$, 
the integrands of 
$III$, $IV$ 
are equal to 
$\cosh (R(s)) \cdot a(s) n^{-1}$ 
multiplied by bounded functions so that 
$III, IV \approx 0$.
Hence 
it suffices to compute the 2nd term 
$II$ 
which has the following form : 
\beq
II
=
\frac {1}{2 \kappa_0} Re [ \Delta J ]
&=&
\frac {1}{2 \kappa_0} Re [
J(\kappa_c) - J(\kappa_d) ]
\\
\mbox{where}
\quad
J(\kappa)
&:=&
\int_0^{nt}
\frac {1}{\sin (\Delta \theta_s)}
e^{2i \theta_s(\kappa)}
a(s) F(X_s) ds.
\eeq
In order to compute 
$J(\kappa)$
we introduce
\beq
J(k ; j ; H)(\kappa)
&:=&
\int_0^{nt}
\frac {1}{\sin (\Delta \theta_s)}
e^{2i k \theta_s(\kappa)}
a(s)^j 
H(X_s) ds
\eeq
for 
$k \in {\bf Z}$,
$j \ge 1$, 
and 
$H \in C^{\infty}(M)$.
By Proposition \ref{J}(1)
we have 
\begin{eqnarray}
\Delta J
&=&
\Delta J(1 ; 1 ;  F)
\nonumber
\\
&\approx&
\frac {1}{\kappa_0}
\langle F\cdot g_{\kappa_0} \rangle
\int_0^{nt}
\cos(\Delta \theta_s)
a(s)^{2} 
ds
\nonumber
\\
&&
- 
\frac {2i}{2 \kappa_0}
\Biggl\{
\frac 12
\Delta J(2 ; 2 ; F g_{\kappa_0})
-
\Delta J(1 ; 2 ; F g_{\kappa_0})
\Biggr\}
+ N_t
\label{DeltaJ}
\end{eqnarray}
where we set  
$g_{\kappa_0} := R_{2\kappa_0} F$.
$N$ 
is a martingale such that 
\beq
\langle N, N \rangle_t
&=&
o \left(
\int_0^{nt} a(s)^2 ds
\right)
\\
\langle N, \overline{N} \rangle_t
&=&
4 \langle \psi \rangle
\int_0^{nt} a(s)^{2} ds
(1 + o(1)), 
\quad
\psi := 
[
g_{\kappa_0}, \overline{g_{\kappa_0}}
]
\eeq
as 
$n \to \infty$.
By (\ref{logtan}), 
the 1st term of RHS in eq.(\ref{DeltaJ}) is equal to the 2nd term of RHS in eq.(\ref{Ricattiequation}).
For the 2nd term 
of RHS in eq.(\ref{DeltaJ}), we use 
Theorem \ref{J}(2).
Noting that 
$J(0 ; j ; H)$ 
is independent of 
$\kappa$ 
so that 
$\Delta J(0 ; j ; H) = 0$, 
we can repeatedly use Theorem \ref{J}(2) for 
$(j_0 - 1)$ - times 
to obtain the sum of negligible terms of the form : 
$\Delta J(k ; j_0 ; H) \approx 0$.
Therefore
\beq
\Delta J(2 ; 2 ; F g_{\kappa_0})
\approx 0,
\quad
\Delta J(1 ; 2 ; F g_{\kappa_0})
\approx 0.
\eeq
Set 
$M$ 
to be the sum of 
$(2 \kappa_0)^{-1} Re N$ 
and all other martingales appeared in the above argument, after taking real part and multiplying 
$(2 \kappa_0)^{-1}$.
Then 
$M$ 
satisfies 
eq.(\ref{M}).
\QED
\end{proof}
%

%
\section{A comparison argument}
In this section
we consider 
$\widetilde{R} := R - e^{(n)}$,  
carry out scaling and time-change, and bound from above and below by the diffusions 
$R_{\pm}$ 
which obey simple SDE's (\ref{Rplus}), (\ref{Rminus}). 
We first prepare some notations. 
Let 
\beq
\widetilde{R}(n t)
&:=&
R(n t) - e^{(n)} (t).
\eeq
$e^{(n)}(t)$ 
is an error term appeared in Proposition \ref{Ricatti}.
Moreover set 
\beq
\gamma &:=& \frac {1}{1 - 2 \alpha} > 1, 
\\
\delta &=& C n^{-\alpha}, 
\quad
\epsilon = C n^{-\beta}, 
\\
\beta 
&:=&
\min \{ \alpha, j_0\alpha-1 \}
=
j_0 \alpha-1, 
\quad
C > 0,
\\
\cosh_+ (r)
&:=&
\sup_{|s-r| < \delta} \cosh s,
\quad
\cosh_- (r)
:=
\inf_{|s-r| < \delta} \cosh s
\\
\tanh_+ (r)
&:=&
\sup_{|s-r| < \delta} \tanh s,
\quad
\tanh_- (r)
:=
\inf_{|s-r| < \delta} \tanh s
\\
\tanh_{+, \epsilon}(r)
&:=&
\left\{
\begin{array}{cc}
(1+\epsilon) \tanh_+ (r) & (r>-\delta) \\
(1-\epsilon) \tanh_+ (r) & (r<-\delta)
\end{array}
\right.
\\
\tanh_{-, \epsilon}(r)
&:=&
\left\{
\begin{array}{cc}
(1-\epsilon) \tanh_- (r) & (r>\delta) \\
(1+\epsilon) \tanh_- (r) & (r<\delta)
\end{array}
\right.
\\
C_n &:=&
\frac {1}{\kappa_0}
\left(
\frac {\langle \psi_{\kappa_0} \rangle}{2}
\right)^{1/2}
\gamma^{\frac 12}
n^{\frac {1}{2 \gamma}}.
\eeq
We consider diffusions  
$R_{\pm}$ 
which are the solutions to 
\begin{eqnarray}
d R_+
&=&
\left(
\lambda(1 + \epsilon) 
\cosh_+ R_+ 
\gamma t^{\gamma-1} 
+
\frac {C_n^2}{2}
\tanh_{+, \epsilon} R_+ 
\right) dt 
+
C_n d W_t
\label{Rplus}
\\
d R_-
&=&
\left(
\lambda(1 - \epsilon) 
\cosh_- R_- 
\gamma t^{\gamma-1} 
+
\frac {C_n^2}{2}
\tanh_{-, \epsilon} R_- 
\right) dt 
+
C_n d W_t
\label{Rminus}
\end{eqnarray}
where 
$W_t$ 
is a standard Brownian motion starting at 
$0$.
Then 
we have a following bound on 
$\widetilde{R}$.
%
\begin{proposition}
\label{comparison}
There is a time change 
$\tau(t)$ 
with 
\[
\tau'(t) = 1 + o(1), 
\quad
n \to \infty
\]
uniformly with respect to 
$\omega \in \Omega$ 
such that 
\begin{equation}
R_{-}(t)\le \widetilde{R}(n \tau(t)^{\gamma}) \le R_{+}(t)
\label{sandwich}
\end{equation}
provided 
the initial values coincide.
\end{proposition}
%
\begin{proof}
We consider 
$R(n t^{\gamma})$ 
instead of 
$R(nt)$ 
and change variables : 
$s = n v^{\gamma}$ 
in eq.(\ref{Ricattiequation}).
\beq
R(n t^{\gamma})
&=&
\lambda
\int_0^{t}
\cosh (R(n v^{\gamma})) 
\cdot
\gamma v^{\gamma-1} dv
\\
&& + 
\frac {1}{2 \kappa_0} Re
\left(
- \frac {\langle F g_{\kappa_0} \rangle}{\kappa_0}
\right)
\int_0^{t} 
n a(n v^{\gamma})^2 
\tanh (R(n v^{\gamma}))
\cdot
\gamma v^{\gamma-1} dv
\\
&& 
+ M_{n t^{\gamma}}
+ E(n t^{\gamma})
\\
d\langle M, M \rangle_{n t^{\gamma}}
&=&
\left(
\frac {1}{2\kappa_0}
\right)^2
\cdot
2 \langle \psi_{\kappa_0} \rangle
\cdot
n a(n t^{\gamma})^2 
\cdot
\gamma 
t^{\gamma-1} 
(1+ o(1))
\,
dt, 
\quad
n \to \infty.
\eeq
We note 
$\langle \psi_{\kappa_0} \rangle
=
-2 Re \langle F g_{\kappa_0} \rangle$
and let 
\beq
D_n &:=&
\frac {1}{\kappa_0}
\left(
\frac {\langle \psi_{\kappa_0} \rangle}{2}
\right)^{1/2}, 
\quad
C_n := 
D_n 
\left(
\gamma n^{1 - 2 \alpha}
\right)^{1/2}
=
D_n 
\gamma^{\frac 12}
n^{\frac {1}{2 \gamma}}.
\eeq
Then 
\beq
R(n t^{\gamma})
&=&
\lambda
\int_0^{t}
\cosh (R(n v^{\gamma}))
\gamma v^{\gamma-1} dv
\\
&& + 
\frac {D_n^2}{2}
\int_0^{t}
\tanh (R(n v^{\gamma}))
\cdot
n a(n v^{\gamma})^2 
\cdot
\gamma 
v^{\gamma-1} dv
+ M_{n t^{\gamma}}
+ E(nt^{\gamma})
\\
d\langle M, M \rangle_{nt^{\gamma}}
&=&
C_n^2  
(1+ o(1))\,
dt, 
\quad
n \to \infty.
\eeq
Let 
$N_t := M_{n t^{\gamma}}/C_n$ 
and take
\beq
\tau(t) :=
\inf \left\{
s \, | \, \langle N \rangle_s > t 
\right\}.
\eeq
Then 
$W_t := N_{\tau(t)}$ 
is a Brownian motion,  
$\tau'(t) 
\stackrel{n \to \infty}{\to} 1 + o(1)$ 
uniformly with respect to 
$\omega \in \Omega$,  
and 
\beq
R(n \tau(t)^{\gamma})
&=&
\lambda
\int_0^{\tau(t)}
\cosh (R(n v^{\gamma}))
\gamma v^{\gamma-1} dv
\\
&& + 
\frac {D_n^2}{2}
\int_0^{\tau(t)}
\tanh (R(n v^{\gamma}))
\cdot
n a(n v^{\gamma})^2 
\cdot
\gamma v^{\gamma-1}\, dv
+ C_n W_t
+ E(n \tau(t)^{\gamma}).
\eeq
Let 
\beq
\widetilde{R}(n t)
&:=&
R(n t) - e^{(n)} (t), 
\quad
\widetilde{E}(n t)
:=
E(n t) - e^{(n)} (t).
\eeq
Then 
\begin{eqnarray}
\widetilde{R}(n \tau(t)^{\gamma})
&=&
\lambda
\int_0^{\tau(t)}
\cosh 
\left(
\widetilde{R}(n v^{\gamma}) 
+ 
e^{(n)}(v^{\gamma})
\right)
\gamma v^{\gamma-1} dv
\nonumber
\\
&& + 
\frac {D_n^2}{2}
\int_0^{\tau(t)}
\tanh 
\left(
\widetilde{R} (n v^{\gamma}) 
+  
e^{(n)}(v^{\gamma}) 
\right)
\cdot
n a(n v^{\gamma})^2
\cdot 
\gamma v^{\gamma-1} dv
\nonumber
\\
&& \qquad
+ C_n W_t
+ \widetilde{E}(n \tau(t)^{\gamma})
+ C.
\label{Rtilde}
\end{eqnarray}
Take 
$t_0 > 0$ 
small enough.
The 
contribution from 
$\widetilde{E}(nt^{\gamma})$ 
for 
$t \le t_0$ 
is bounded which we ignore. 
For 
$t \ge t_0$, 
\beq
\widetilde{E}(n t^{\gamma})
&=&
\int_0^{t}
\cosh \left(
\widetilde{R}(n v^{\gamma}) + 
e^{(n)}(v^{\gamma})
\right)
b(n v^{\gamma}) c_1 (n v^{\gamma}) 
n \gamma v^{\gamma-1} \, dv
\\
&&+
\int_0^{t}
\tanh \left(
\widetilde{R}(n v^{\gamma}) + 
e^{(n)}(n v^{\gamma})
\right)
a(n v^{\gamma})^3 c_2 (n v^{\gamma}) 
n \gamma v^{\gamma-1} \, dv
\\
d \tilde{E}(n t^{\gamma})
& \le &
\cosh (R(n t^{\gamma})) 
\left\{
\frac 1n a(n t^{\gamma})
+
a'(n t^{\gamma})
+
a(n t^{\gamma})^{j_0}
\right\}
c_1(n t^{\gamma}) 
n \gamma t^{\gamma-1} d t
\\
&& \qquad
+
| \tanh (R(n t^{\gamma})) |
a(n t^{\gamma})^3
c_2(n t^{\gamma}) 
n  \gamma t^{\gamma-1} dt
\\
& \le &
C 
\cosh (R(n t^{\gamma}))
\left(
\frac 1n
\cdot
(n t^{\gamma})^{-\alpha}
+
(n t^{\gamma})^{-\alpha-1}
+
(n t^{\gamma})^{-\alpha j_0}
\right)
n \gamma  t^{\gamma-1} dt
\\
&& \qquad+
C 
| \tanh (R(n t^{\gamma})) |
(n t^{\gamma})^{-3 \alpha}
n \gamma  t^{\gamma-1} dt
\\
& \le &
C n^{- \beta} \cosh (R(n t^{\gamma})) 
t^{\gamma-1} dt
+
O (n^{1 - 3 \alpha})
| \tanh (R_{n t^{\gamma}}) |
t^{(1 - 3\alpha)\gamma-1}  dt
\eeq
where 
$\beta 
:=
\min \{ \alpha, j_0\alpha-1 \}
=
j_0 \alpha-1$.
Thus in eq.(\ref{Rtilde}),  
$\widetilde{E}(n \tau(t)^{\gamma})$
is lower order compared to the 1st and the 2nd terms, and then by the comparison theorem, we have
\[
R_{-}(t)\le \widetilde{R}(n \tau(t)^{\gamma}) \le R_{+}(t).
\]
\QED
\end{proof}
%
\section{Allez-Dumaz analysis}
In this section, 
we show, along the argument in \cite{AD}, that 
(i)
the marginal 
$\xi_L (I)$ 
($I = [\lambda_1, \lambda_2]$) 
of 
$\xi_L$ 
converges to Poisson distribution, 
and 
(ii)
the joint limit of 
$\xi_L(I_1), \cdots, \xi_L(I_N)$ 
are independent.

Propositions and lemmas 
in this section can be proved in the same manner as in 
\cite{AD} 
by putting 
$\beta = n^{-\frac {1}{\gamma}}$, 
but we give proofs of them in Appendix II 
for the sake of completeness.
%
\subsection{Preliminary : explosion time of stationary approximation}
In this subsection 
we study the explosion time of the stationary approximation 
$S_{\pm}$ 
of 
$R_{\pm}$
which are the solution to another SDE's (\ref{S}) 
where the coefficient 
$\gamma t^{\gamma-1}$ 
in the drift term in eq.(\ref{Rplus}), (\ref{Rminus}) 
are replaced by 
$1$ : 
\begin{equation}
d S_{\pm}
=
\left(
\lambda (1 \pm \epsilon) 
\cosh_{\pm} (S_{\pm})
+
\frac {C_n^2}{2}
\tanh_{\pm, \epsilon} (S_{\pm})
\right) dt
+ 
C_n d W_t.
\label{S}
\end{equation}
If 
$| S_{\pm} | > \delta$, 
the drift term of these SDE's are just 
the constant multiples of the shift of 
$\cosh$, $\tanh$, so that the analysis in 
\cite{AD} 
also works.
Because the potential 
corresponding to the drift term in SDE 
(\ref{S})  
has a barrier between the local minimum in the well and the local maximum, 
we have a ``memory-loss effect" so that 
the explosion time converges to the exponential distribution. 
More precisely, let 
$\zeta_{\pm}$ 
be the explosion time of 
$S_{\pm}$ 
and let 
\beq
t_{n}^{(\pm)}(r) &:=& {\bf E} [ \zeta_{\pm} | S_{\pm}(0) = r]
\\
g_{n}^{(\pm)}(r) &:=& 
{\bf E} [ 
e^{
-\xi 
\cdot \frac {\lambda}{\pi} \cdot
\zeta_{\pm}
}
| S_{\pm}(0) = r]
\eeq
be the expectation value and the Laplace transform of 
$\zeta_{\pm}$ 
conditioned 
$S_{\pm}(0) = r$ 
respectively. 
We then have
\begin{proposition}
\label{exponential}
\beq
&&
\lim_{r \downarrow - \infty}
\lim_{n \to \infty}
t_{n}^{(\pm)}(r) 
= 
\frac {\pi}{\lambda}
\\
&&
\lim_{r \downarrow -\infty}
\lim_{n \to \infty}
g_{n}^{(\pm)}(r) = \frac {1}{1+\xi}.
\eeq
\end{proposition}
%

\subsection{Poisson convergence for marginals}
In this subsection, 
we prove that the marginal 
$\xi_L(I)$ 
of 
$\xi_L$ 
on an interval 
$I$ 
converges to a Poisson distribution by 
showing that the jump points of the function 
$t \mapsto 
\left\lfloor 
\Theta_{n \tau(t)^{\gamma}} 
\right\rfloor$ 
converges to a Poisson process.
This will be done 
by dividing the time interval
$[0,1]$ 
into small random ones 
$I_k$ 
and approximating 
$R_{\pm}$ 
by 
$S_{\pm}$ 
on each 
$I_k$'s.
In order that 
such approximation work, 
we need to show that  
$\{ \Theta_{n \tau(t)^{\gamma}}(\lambda) \}_{\pi}$ 
is sufficiently small on sufficiently large portion of the time interval, which is guaranteed by Lemma \ref{Xi}. 
In order to prove 
Lemma \ref{Xi}, 
we need some estimates on the explosion time for 
\[
R^{(n)}(t) := \widetilde{R}(n \tau(t)^{\gamma})
\]
which are done in Lemmas \ref{explosion1}, \ref{explosion2}.
Lemmas \ref{explosion3}, \ref{indicator} 
are rephrase of Lemmas \ref{explosion1}, \ref{Xi} 
respectively.
Since 
$\tau'(t) = 1+o(1)$ 
uniformly in 
$\omega \in \Omega$, 
all statements in this subsection are also valid for 
$\widetilde{R}(n t^{\gamma})$.
Let 
\[
T_r  := 
\inf \left\{
s \; \middle| \; 
R^{(n)}(s) = r 
\right\}
\]
be the hitting time of 
$R^{(n)}$ 
to 
$r \in {\bf R} \cup \{ + \infty \}$. 
We denote by 
${\bf P}_{r_0,\, t_0}$ 
the law of 
$R^{(n)}$ 
conditioned 
$R^{(n)}(t_0) = r_0$. 
If 
$t_0 = 0$, 
we simply write 
${\bf P}_{r_0, \, t_0} = {\bf P}_{r_0}$. 
\begin{lemma}
\label{explosion1}
Let 
$0 < \epsilon < 1$, 
$c > \gamma + \frac 12$. 
Then 
we can find a constant 
$c' > 0$ 
such that 
\beq
{\bf P}_{
\epsilon 
\log 
n^{\frac {1}{\gamma}}
}
\left(
T_{+\infty} <
\frac {5c}{C_n^2} 
\log n^{\frac {1}{\gamma}}
\right)
\ge 
1 - 
n^{- \frac {c'}{\gamma}}.
\eeq
\end{lemma}
{\it Idea of proof} : 
(i)
we derive the probability of the event that 
$R^{(n)}$ 
reaches 
$c \log n^{\frac {1}{\gamma}}$
before hitting 
$\frac {\epsilon}{2} \log n^{\frac {1}{\gamma}}$, 
by the time
$\frac {4c}{C_n^2} 
\log n^{\frac {1}{\gamma}}$.
Since 
the drift term is bounded from below by 
$\frac 14 C_n^2 dt$,  
this is possible 
provided the Brownian motion term satisfies 
$C_n 
\inf
\{ W_t \, | \, 
0 \le t \le \frac {4c}{C_n^2} \log n^{\frac {1}{\gamma}}
\}
\ge - \frac {\epsilon}{2} \log n^{\frac {1}{\gamma}}$
which happens with probability 
$\ge 1 - 
n^{
- \frac {c'}{\gamma}
}$.
(ii)
Once
$R^{(n)}$
reaches 
$c \log n^{\frac {1}{\gamma}}$, 
it explodes by the time 
$\frac {c}{C_n^2} 
\log n^{\frac {1}{\gamma}}$
which 
can be proved by studying the explosion time of an ODE explicitly. 
\QED
\begin{lemma}
\label{explosion2}
\beq
{\bf P}_{
- \frac 14 
\log n^{\frac {1}{\gamma}}
}
\left(
T_{+\infty}
<
\frac {5c+1}{C_n^2}
\log n^{\frac {1}{\gamma}}
\right)
\ge
n^{- \frac {1}{2\gamma}}.
\eeq
\end{lemma}
{\it Idea of Proof}
 : 
on account of Lemma \ref{explosion1} with 
$\epsilon = 1/4$, 
it is sufficient to estimate the probability 
${\bf P}_{-\frac 14 \log n^{\frac {1}{\gamma}}}
\left(
T_{\frac 14 \log n^{\frac {1}{\gamma}}}
<
\frac {1}{C_n^2} 
\log n^{\frac {1}{\gamma}}
\right)$
which can be done similarly by the  idea (i)  for Lemma \ref{explosion1}.
\QED
\begin{lemma}
\label{Xi}
Let 
\beq
\Xi_n(t)
&:=&
{\bf E}_{- \infty}
\left[
\int_0^t
1 \left(
R^{(n)}(u) \ge  
- \frac 14 
\log n^{\frac {1}{\gamma}}
\right) du
\right].
\eeq
Then 
we can find a constant 
$C$ 
such that 
\beq
\Xi_n (t)
& \le &
C
n^{- \frac {1}{2\gamma}}
\log n^{\frac {1}{\gamma}}.
\eeq
\end{lemma}
{\it Idea of Proof} : 
by Lemma \ref{explosion2}, 
if 
$R^{(n)}(u) \ge  
- \frac 14 
\log n^{\frac {1}{\gamma}}$, 
we have 
$T_{+\infty} < 
\frac {5c+1}{C_n^2} 
\log n^{\frac {1}{\gamma}}$, 
that is, 
it will explode by the time 
$\frac {5c+1}{C_n^2} 
\log n^{\frac {1}{\gamma}}$, 
with a good probability.
Hence 
the quantity inside the expectation in the definition of 
$\Xi_n(t)$ 
is bounded from above by the number of explosions multiplied by  
$\frac {5c+1}{C_n^2} 
\log n^{\frac {1}{\gamma}}$.
On the other hand, 
the expectation value of the number of explosions is bounded from above. 
\QED

We 
shall study the distribution of the jump points of the function 
$t \mapsto \left
\lfloor
\Theta_{n \tau(t)^{\gamma}}(\lambda) / \pi
\right\rfloor$.
The corresponding 
point process is defined by 
\beq
\tilde{\mu}_{\lambda}^{(n)}
&:=&
\sum_k
\delta_{
\tilde{\zeta}_k^{\lambda}
}
\\
\mbox{where }
\quad
\tilde{\zeta}_k^{\lambda}
&:=&
\inf
\left\{
t \in [0,1]\; \middle| \;
\Theta_{n \tau(t)^{\gamma}}(\lambda)
\ge 
k \pi
\right\}.
\eeq
Then 
the statements of 
Lemma \ref{explosion1}, \ref{Xi}
have the following form.
\begin{lemma}
\label{explosion3}
Let 
$0 < \epsilon < 1$ 
$c > \gamma + \frac 12$.
Then 
conditioned on 
$\{ \Theta_0(\lambda) \}_{\pi}
=
\pi - 2 \arctan 
n^{- \frac {\epsilon}{\gamma}}$, 
we have 
\beq
{\bf P}
\left(
\tilde{\zeta}_1^{\lambda}
<
\frac {5c}{C_n^2}
\log n^{\frac {1}{\gamma}}
\right)
\ge
1 - 
n^{- \frac {c'}{\gamma}}.
\eeq
\end{lemma}
\begin{lemma}
\label{indicator}
Let 
\beq
\Xi_n (t)
:=
{\bf E}
\left[
\int_0^t 
1 \left(
\Theta_{n \tau(u)^{\gamma}}(\lambda)
\ge
2 \arctan
n^{- \frac {1}{4\gamma}}
\right)
du
\right].
\eeq
Then 
we can find a constant 
$C$ 
such that 
\beq
\Xi_n (t)
\le
C
n^{- \frac {1}{2\gamma}}
\log n^{\frac {1}{\gamma}}.
\eeq
\end{lemma}
We can 
now prove that the jump points of the function 
$t \mapsto \left
\lfloor
\Theta_{n \tau(t)^{\gamma}}(\lambda) / \pi
\right\rfloor$
converges to a Poisson process.
\begin{proposition}
\label{marginal}
\beq
\tilde{\mu}_{\lambda}^{(n)}
\stackrel{d}{\to}
\mbox{Poisson}
\left(
\frac {\lambda}{\pi}
\gamma t^{\gamma-1}  
1_{[0,1]}(t) dt
\right)
\eeq
and the same statement also holds for the point process 
$\mu_{\lambda}^{(n)}$
whose atoms consist of 
\beq
\zeta_k^{\lambda}
&:=&
\inf
\left\{
t \in [0,1] \; \middle| \;
\Theta_{n t^{\gamma}}(\lambda)
\ge 
k \pi
\right\}.
\eeq
\end{proposition}
{\it Idea of Proof} : 
Let 
\beq
I_k 
&:=&
\left[
\frac {T_k}{N}, \frac {T_{k+1}}{N}
\right]
\quad
\mbox{ where }
\quad
T_k 
:=
\sum_{i=1}^k \tau_i, 
\quad
\tau_i = 
\mbox{unif}
\left(
\frac 12, \frac 32
\right).
\eeq
Let 
$S_{\pm}^{(n)}$
be the solution to the following SDE's 
where the constant 
$\lambda$
in SDE (\ref{S}) is replaced by 
$\gamma \left( \frac {T_{k+1}}{N} \right)^{\gamma-1}$, 
$\gamma \left( \frac {T_{k}}{N} \right)^{\gamma-1}$
respectively : 
\beq
d S_{\pm}^{(n)}
&=&
\left(
\lambda_k^{\pm} 
(1 \pm \epsilon) 
\cosh_{\pm} (S_{\pm}^{(n)})
+
\frac {C_n^2}{2}
\tanh_{\pm, \epsilon} (S_{\pm}^{(n)})
\right) dt
+ 
C_n d B_t, 
\quad
t \in I_k
\\
\mbox{where}&&
\lambda_k^+
=
\gamma \left( \frac {T_{k+1}}{N} \right)^{\gamma-1}, 
\quad
\lambda_k^-
=
\gamma \left( \frac {T_{k}}{N} \right)^{\gamma-1}
\eeq
with initial values 
$S_{\pm}^{(n)}
\left(
\frac {T_k}{N}
\right)
:=
R^{(n)}
\left( \frac {T_k}{N} \right)
$
on each 
$I_k$.
We remark 
that, once 
$S_{\pm}^{(n)}$ 
explode to 
$+ \infty$, 
it starts at 
$- \infty$ 
again and so on.
Let 
$\Theta_{\pm}^{(n)}$
defined by 
\beq
S_{\pm}^{(n)}
&=&
\log \tan \frac {\Theta_{\pm}^{(n)}}{2}, 
\eeq
in other words, 
$\Theta_{\pm}^{(n)}
:=
2 \arctan e^{S_{\pm}^{(n)}}$.
Then by
eq.(\ref{sandwich}) 
and using comparison theorem between 
$S_{\pm}$ and $R_{\pm}$, 
\beq
\Theta_{-, t}^{(n)} (\lambda) 
\le
\Theta_{n \tau(t)^{\gamma}} 
(\lambda)
\le 
\Theta_{+, t}^{(n)} (\lambda).
\eeq
Thus 
we can estimate the number of jump points of 
$\left
\lfloor
\Theta_{n \tau(t)^{\gamma}}(\lambda) / \pi
\right\rfloor$
from above and below by those of 
$\left
\lfloor
\Theta_{\pm, t}^{(n)}(\lambda) / \pi
\right\rfloor$.
By Lemma 
\ref{indicator} 
and by the definition of 
$T_k$, 
on each starting point of the interval 
$I_k$, 
we can suppose 
$\Theta_{n \tau(t)^{\gamma}}(\lambda)
\le
2 \arctan 
n^{- \frac {1}{4\gamma}}$
with a good probability, so that by Proposition 
\ref{exponential}, 
the explosion time of 
$\Theta_{\pm}^{(n)}$
converges to the exponential distribution on each intervals, which proves the statement of Proposition \ref{marginal} for 
$\Theta_{n \tau(t)^{\gamma}}(\lambda)$.
Since 
$\tau'(t) = 1+o(1)$ 
uniformly in  
$\omega \in \Omega$, 
the same statement also holds for 
$\mu_{\lambda}^{(n)}$.
\QED
\begin{remark}
\label{lambdadash}
Let 
$\lambda < \lambda'$ 
and let 
\beq
\mu_{\lambda, \lambda'}^{(n)}
&:=&
\sum_{k}
\delta_{ \zeta_k^{\lambda, \lambda'} }
\\
\mbox{where }
\quad
\zeta_k^{\lambda, \lambda'}
&:=&
\inf \left\{
t \in [0,1] 
\, | \,
\Theta_{n t^{\gamma}}(\lambda')
-
\Theta_{n t^{\gamma}}(\lambda)
\ge k \pi
\right\}.
\eeq
We can apply 
all the arguments in previous sections for 
$\Theta_{n t^{\gamma}}(\lambda')
-
\Theta_{n t^{\gamma}}(\lambda)$
yielding 
%
\beq
\mu_{\lambda, \lambda'}^{(n)}
\stackrel{d}{\to} 
Poisson
\left(
\frac {\lambda' - \lambda}{\pi}
\gamma t^{\gamma-1} 1_{[0,1]} dt
\right).
\eeq
\end{remark}
%

\subsection{Limiting Coupled Poisson Process}
For 
$0 < \lambda < \lambda'$, 
let 
$P_{\lambda} := \lim_{n \to \infty} 
\mu^{(n)}_{\lambda}$, 
$P_{\lambda'} := \lim_{n \to \infty} 
\mu^{(n)}_{\lambda'}$, 
$P_{\lambda, \lambda'} := \lim_{n \to \infty} 
\mu^{(n)}_{\lambda, \lambda'}$
be the limiting Poisson processes described in Proposition \ref{marginal} and Remark \ref{lambdadash}.
In this subsection, 
we show that 
(i) 
they are realized jointly as 
${\cal F}_t$-Poisson processes under suitable filtration(Lemma \ref{jointlyPoisson}), 
(ii) 
the sets 
${\cal P}_{\lambda}$, 
${\cal P}_{\lambda'}$, 
${\cal P}_{\lambda, \lambda'}$ 
of corresponding atoms satisfy 
${\cal P}_{\lambda} 
\subset 
{\cal P}_{\lambda'}$
(Lemma \ref{monotonicity2}), and 
(iii) 
${\cal P}_{\lambda}
\cap 
{\cal P}_{\lambda, \lambda'} = \emptyset$
(Lemma \ref{independent}).
The independence of 
$P_{\lambda}$, $P_{\lambda, \lambda'}$ 
(and thus independence of finite number of marginals of 
$\xi_L$ 
on intervals) then follows from those observations. 
But 
first of all we need to show that the ``fractional part" of 
$\Theta (\lambda)$, 
$\Theta (\lambda')$ 
also obey the same ordering as 
$\lambda, \lambda'$ 
for sufficiently large portions in time (Lemma \ref{order}). 
We recall 
$\{ x \}_{\pi} := 
x - \left\lfloor x / \pi \right\rfloor \pi$.
\begin{lemma}
\label{order}
Let 
$0 < \lambda < \lambda'$
and 
\beq
\Upsilon_n(t)
&:=&
{\bf E}
\left[
\int_0^t 
1 \left(
\left\{ \Theta_{nu^{\gamma}}(\lambda') \right\}_{\pi}
\le
\left\{ \Theta_{n u^{\gamma}}(\lambda) \right\}_{\pi}
\right)
du
\right]
\eeq
then we can find a constant 
$C$ 
such that 
\beq
\Upsilon_n (t) 
\le
C
n^{- \frac {c'}{\gamma}}.
\eeq
\end{lemma}
{\it Idea of Proof} : let 
\beq
{\cal E}_u
&:=&
\left\{
\left\{ \Theta_{nu^{\gamma}}(\lambda') \right\}_{\pi}
\le
\left\{ \Theta_{n u^{\gamma}}(\lambda) \right\}_{\pi}
\right\}, 
\quad
u \in [0,1]
\\
\zeta_u
&:=&
\sup
\left\{
\zeta^{\lambda'}_k
\, \middle| \, 
\zeta_k^{\lambda'} \le u
\right\}, 
\quad
\zeta_k^{\lambda'}
:=
\inf
\left\{
t \in [0,1] \; | \;
\Theta_{n t^{\gamma}}(\lambda')
\ge 
k \pi
\right\}
\\
u_0 &:=&
u - \frac {5c}{C_n^2} 
\log n^{\frac {1}{\gamma}}, 
\quad
c > \gamma + \frac 12.
\eeq
On the event 
${\cal E}_u$, 
we consider the following three possibilities.
\\
(i)
the latest jump of the function 
$t \mapsto \left\lfloor
\Theta_{n t^{\gamma}}(\lambda')/\pi
\right\rfloor$
before 
$u$ 
occurs after 
$u_0$
\\
(ii)
the latest jump of 
$\left\lfloor
\Theta_{n t^{\gamma}}(\lambda')/\pi
\right\rfloor$
before 
$u$ 
occurs before 
$u_0$,
and 

$\{ \Theta_{nu_0^{\gamma}}(\lambda) \}_{\pi} \le 
2 \arctan
n^{- \frac {1}{4\gamma}}$, \\
(iii)
the latest jump of 
$\left\lfloor
\Theta_{n t^{\gamma}}(\lambda')/\pi
\right\rfloor$
before 
$u$ 
occurs before 
$u_0$,
and 

$\{ \Theta_{nu_0^{\gamma}}(\lambda) \}_{\pi} 
> 
2 \arctan
n^{- \frac {1}{4\gamma}}$.\\
Then 

\noindent
(i) 
the probability of the event 
(i) 
is bounded from above by 
$
n^{- \frac {1}{2\gamma}}
\log n^{\frac {1}{\gamma}}
\cdot
{\bf E}[\mu_{\lambda}^n [0,t]]$.
\\
(ii)
Let 
$\widetilde{\zeta}_{2 \pi}$
be the explosion time of 
$\Theta_{n t^{\gamma}}(\lambda, \lambda'):=
\Theta_{n t^{\gamma}}(\lambda') - \Theta_{n t^{\gamma}}(\lambda)$ 
for which we can carry out the arguments in previous sections.
Then 
in Case (ii) we must have 
$\tilde{\zeta}_{2 \pi} \ge 
\frac {5 c}{C_n^2} \log n^{\frac {1}{\gamma}}$
of which the probability is bounded from above by 
$
n^{- \frac {c'}{\gamma}}$
\\
(iii)
Lemma \ref{indicator}
gives the bound on the probability of Case (iii).
\QED\\
In what follows, we set 
$\lambda < \lambda' < \lambda''$.
Since the set of triples
$\{ 
(\mu_{\lambda}^{(n)},
\mu_{\lambda'}^{(n)}, 
\lambda_{\lambda', \lambda''}^{(n)}), 
\;
n \ge 0
\}$
is tight as a set of Radon measures on 
${\bf R}_+$, 
we can find a subsequence 
$(n_k)$ 
such that 
\beq
(
\mu_{\lambda}^{(n_k)},
\mu_{\lambda'}^{(n_k)}, 
\lambda_{\lambda', \lambda''}^{(n_k)}
)
\to
\left(
P_{\lambda}, 
P_{\lambda'}, 
P_{\lambda', \lambda''}
\right)
\eeq
where 
$P_{\lambda}$, $P_{\lambda'}$, $P_{\lambda', \lambda''}$
are Poisson processes which turn out to be independent of the choice of convergent subsequences.
\begin{lemma}
\label{jointlyPoisson}
Let 
\beq
{\cal F} &:=& ( {\cal F}_t )_{t \ge 0}
\\
{\cal F}_t
&:=&
\sigma 
\left(
P_{\lambda}(s), 
P_{\lambda'}(s), 
P_{\lambda', \lambda''}(s) ;  
\;
0 \le s \le t
\right).
\eeq
Then 
$P_{\lambda}$, $P_{\lambda'}$, $P_{\lambda', \lambda''}$
are the 
$({\cal F}_t)$-Poisson processes whose intensity measures are equal to 
$\pi^{-1} \lambda
\gamma t^{\gamma-1}1_{[0,1]}(t) dt$, 
$\pi^{-1} \lambda'
\gamma t^{\gamma-1}1_{[0,1]}(t) dt$, 
and 
$\pi^{-1}
(\lambda'' - \lambda')
\gamma t^{\gamma-1}1_{[0,1]}(t) dt$ 
respectively.
\end{lemma}
Let 
${\cal P}_{\lambda}$, 
${\cal P}_{\lambda'}$, 
${\cal P}_{\lambda', \lambda''}$
be the set of atoms of 
$P_{\lambda}$, $P_{\lambda'}$, $P_{\lambda', \lambda''}$
respectively.
\begin{lemma}
\label{monotonicity2}
If 
$\lambda < \lambda'$, 
${\cal P}_{\lambda} \subset {\cal P}_{\lambda'}$ 
a.s.
\end{lemma}
{\it Idea of Proof} : 
suppose that 
there are no atoms of 
$\mu_{\lambda'}^{(n)}$ 
near the atom 
$\xi$ of 
$\mu_{\lambda}^{(n)}$ 
for large 
$n$. 
Then 
we should have 
$\{ \Theta_{nt^{\gamma}}(\lambda') \}_{\pi}
<
\{ \Theta_{nt^{\gamma}}(\lambda) \}_{\pi}$
near 
$\xi$ 
of which the probability is estimated from above by 
Lemma \ref{order}.
\QED
\begin{lemma}
\label{independent}
We have 
${\cal P}_{\lambda}
\cap
{\cal P}_{\lambda', \lambda''} = \emptyset$.
Hence 
by Lemma \ref{jointlyPoisson}, 
${\cal P}_{\lambda}$
and 
${\cal P}_{\lambda', \lambda''}$
are independent.
\end{lemma}
{\it Idea of Proof} : 
we shall show 
${\cal P}_{\lambda}
\cap
{\cal P}_{\lambda, \lambda'} = \emptyset$.
Otherwise, 
we can find an atom 
$\xi$ 
of 
$\mu_{\lambda, \lambda'}^{(n)}$ 
near those of 
$\mu_{\lambda}^{(n)}$
for large 
$n$. 
If we have 
$\{ \Theta_{nt^{\gamma}}(\lambda') \}_{\pi}
<
\{ \Theta_{nt^{\gamma}}(\lambda) \}_{\pi}$
near 
$\xi$, 
this probability is estimated from above by 
Lemma \ref{order}.
If, on the contrary, we have 
$\{ \Theta_{nt^{\gamma}}(\lambda') \}_{\pi}
\ge
\{ \Theta_{nt^{\gamma}}(\lambda) \}_{\pi}$, 
then 
$\left\lfloor
\Theta_{nt^{\gamma}}(\lambda') / \pi
\right\rfloor$
jumps twice in a neighborhood of 
$\xi$. 
Since 
the jump points of 
$\left\lfloor
\Theta_{nt^{\gamma}}(\lambda') / \pi
\right\rfloor$
converges to a Poisson process, the probability of such events are relatively small.
\QED
By using these lemmas, we can show 
\begin{proposition}
\label{lambdamarginal}
Let 
$\nu^{(n)}$ 
be a point process on 
${\bf R}$ 
defined by 
\beq
\nu^{(n)}[\lambda_1, \lambda_2]
&=&
\left\lfloor
\frac {
\Theta_{n }(\lambda_2)
-
\Theta_{n }(\lambda_1)
}
{\pi}
\right\rfloor
\eeq
then 
\[
\nu^{(n)} 
\stackrel{d}{\to} Poisson 
\left( \frac {d \lambda}{\pi} \right).
\]
\end{proposition}
%
\section{Proof of Theorems}
\subsection{Proof of Theorem 2}
The first 
statement (1) of Theorem \ref{main} 
can proved in the same manner as   \cite{KN} Proposition 7.1 :
the only major difference is to show 
\beq
\lim_{t \to \infty}
\int_1^t 
s^{- 3 \alpha}
\exp 
\left(
- \int_s^t u^{-2 \alpha} du 
\right)
ds 
= 0
\eeq
which is straightforward. 
For 
the second statement (2) of Theorem \ref{main}, 
we summarize the facts obtained in previous sections. 
\\

\noindent
(1)
Let 
\beq
\zeta^{(n)}(\lambda) &:=& \sum_j \delta_{\tau_j^{(n)}(\lambda)}
\\
\mbox{where }
\quad
\tau_j^{(n)}(\lambda) &:=&
\inf \left\{
t \in [0,1] \, | \, 
\Theta_{nt}(\lambda) = j \pi
\right\}.
\eeq
Then by Proposition \ref{marginal} 
\[
\zeta^{(n)}(\lambda)
\to
Q_{\lambda} := Poisson
\left(
\frac{\lambda}{\pi}
1_{[0,1]} dt
\right).
\]
In other words, the function 
$t \mapsto \left\lfloor
\Theta_{nt}(\lambda)/ \pi
\right\rfloor$
converges to a Poisson jump process.
\\
(2)
By Proposition \ref{Thetalimit}, 
${\bf E}[ 
\Theta_{n t}(\lambda)
]
\to
\lambda t$.
\\
(3)
For 
$0 < \lambda < \lambda'$, 
let 
\beq
\zeta^{(n)}(\lambda, \lambda')
&=&
\sum_j \delta_{\tau_j^{(n)}(\lambda, \,\lambda')}
\\
\mbox{where }
\quad
\tau_j^{(n)}(\lambda, \lambda') 
&:=&
\inf \left\{
t \in [0,1] \, | \, 
\Theta_{nt}(\lambda')
-
\Theta_{nt}(\lambda) 
= j \pi
\right\}.
\eeq
Then
\beq
\zeta^{(n)}(\lambda, \lambda')
\to
Q_{\lambda, \lambda'} := Poisson
\left(
\frac{\lambda' - \lambda}{\pi}
1_{[0,1]} dt
\right)
\eeq
and 
$Q_{\lambda}$ 
and 
$Q_{\lambda, \lambda'}$
are independent.\\

\noindent
By (1), (2), we have 
\beq
{\bf E}
\left[
\left\lfloor
\frac {
\Theta_{nt}(\lambda)
}
{\pi}
\right\rfloor
\right]
\to
\frac {\lambda}{\pi} 
t^{}, 
\quad
{\bf E}
\left[
\frac {\Theta_{nt}(\lambda)}{\pi}
\right]
\to
\frac {\lambda}{\pi} 
t^{}
\eeq
so that, writing 
\beq
\frac {\Theta_{\lambda}(t)}{\pi}
&=&
\left\lfloor
\frac {\Theta_{\lambda}(t)}{\pi}
\right\rfloor
+
\epsilon_t^{(n)}, 
\quad
\epsilon_t^{(n)} \ge 0
\eeq
we have 
${\bf E}[ \epsilon_t^{(n)} ] \to 0$ 
which implies 
$\epsilon_t^{(n)} \to 0$
in probability
\footnote{
In \cite{VV}, 
they showed that, for 
$\beta \le 2$, 
$\alpha_t(\lambda)$ 
converges to 
$\alpha_{\infty}(\lambda)$ 
from above which is consistent with this argument.
}.
It follows that 
$t \mapsto 
\Theta_{nt}(\lambda)/ \pi$
also converges to a Poisson jump process, and in particular, 
\beq
\widehat{\Theta}_t(\lambda)
:=
\lim_{n \to \infty}
\Theta_{n t^{}}(\lambda)
\eeq
takes values in 
$\pi {\bf Z}$
for a.e.
$t$.
Moreover, 
by Remark \ref{monotonicity1} and 
Lemma \ref{monotonicity2}, 
$\widehat{\Theta}_t(\lambda)$ 
is non-decreasing with respect to  
$(t, \lambda)$, 
so that it is a distribution function of a point process 
$\eta$
on 
${\bf R}^2$ 
whose marginals on rectangles have Poisson distribution.
Let 
\beq
N(t_1, t_2 ; \lambda_1, \lambda_2)
&=&
\left(
\widehat{\Theta}_{t_2}(\lambda_2) - \widehat{\Theta}_{t_1}(\lambda_2)
\right)
-
\left(
\widehat{\Theta}_{t_2}(\lambda_1) - \widehat{\Theta}_{t_1}(\lambda_1)
\right)
\eeq
be the number of atoms of 
$\eta$ 
in 
$[t_1, t_2] \times [\lambda_1, \lambda_2]$.
By Lemma \ref{independent}, 
\beq
N(t_1, t'_1 ; \lambda_1, \lambda'_1), 
\cdots, 
N(t_n, t'_n ; \lambda_n, \lambda'_n)
\eeq
are independent obeying 
$Poisson \left(
\pi^{-1}
(\lambda'_j - \lambda_j)
\left( t'_j - t_j \right)
\right)$, 
$j=1, 2, \cdots, n$
which proves the statement (2) of Theorem \ref{main}.
%
\subsection{Proof of Theorem \ref{Poisson}}
By Proposition \ref{lambdamarginal}, 
we have 
\beq
( \Theta_n (c_i) - \Theta_n (d_i), 
i=1, \cdots, k)
\stackrel{d}{\to}
( \widehat{\Theta}_1(c_i) - \widehat{\Theta}_1(d_i), 
i=1, \cdots, k)
\eeq
for any 
$k \in {\bf N}$, $c_i$, $d_i \in {\bf R}$
and 
$\widehat{\Theta}_1(\cdot)$
is a Poisson jump process.
By 
\cite{KN} Lemma 9.1, 
\[
\Theta_n(\cdot) \stackrel{d}{\to} \widehat{\Theta}_1(\cdot)
\]
as a non-decreasing function valued process.
By Skorohod's theorem, 
we may suppose that 
\[
\Theta_n(c) \to \widehat{\Theta}_1(c), 
\quad
a.s.
\]
at any continuity point of 
$\widehat{\Theta}_1(c)$.
Fix 
a.s.
$\omega \in \Omega$, 
$K \in {\bf N}$, 
$\epsilon > 0$
and let 
$\tau_1, \tau_2, \cdots$
be the jump points of 
$\widehat{\Theta}_1(\cdot)$.
Then for large 
$n$, 
\beq
&&
| \Theta_n(\tau_k - \epsilon) - (k-1) \pi| < \epsilon
\\
&&
| \Theta_n(\tau_k+ \epsilon) - k \pi | < \epsilon, 
\quad
k = 1, 2, \cdots, K.
\eeq
By the monotonicity of 
$\Theta_n(\cdot)$, 
if 
$\Theta_n(\tau_k-\epsilon) < y < \Theta_n(\tau_k + \epsilon)$, 
we have 
\beq
| ( \Theta_n )^{-1} (y) - \tau_k | < \epsilon
\eeq
so that, if 
$(k-1) \pi + \epsilon < y < k \pi - \epsilon$, 
we have
\beq
| ( \Theta_n )^{-1}(y) - \tau_k | < \epsilon.
\eeq
Let 
$\Xi (y)$ 
be the inverse of the Poisson jump process
$\widehat{\Theta}_1(\cdot)$
(it may be set  
to take arbitrary values at the discontinuity points). 
%
%
Since 
$\widehat{\phi}_t$
is uniformly distributed on 
$[0, \pi)$, 
its distribution never have a atom at 
$0$ 
so that, taking 
$n \to \infty$ 
in (\ref{Laplace}), we have 
\beq
{\bf E}[ e^{- \xi_L (f)} ]
\to 
{\bf E}\left[
\exp 
\left(
- \sum_{n \in {\bf Z}}
f 
\left(
\Xi(n \pi + \theta)
\right)
\right)
\right]
=
{\bf E}[ e^{ - \zeta_P (f)} ]
\eeq
where 
$\zeta_P = Poisson (\pi^{-1} d \lambda)$.
%

\section{Appendix I}
In this section 
we prepare some estimates necessary to prove Proposition \ref{Ricatti}.
The basic strategy 
of our computation is that, for the terms whose integrand contains a factor of the form 
$e^{i \kappa s} H(X_s) ds$
$(\kappa \ne 0)$, 
we use eq.(\ref{Itoformula}) 
and perform the integration by parts to obtain the terms whose integrands are multiplied by 
$a(s)$ 
or 
$a'(s)$ 
so that they have better decay. 
We may 
continue this process as many times we need to finally obtain the negligible terms. 
On the other hand, 
for the terms with 
$H(X_s) ds$
(that is, 
$\kappa = 0$), 
we use eq.(\ref{SecondItoformula})
instead to obtain the 2nd term of RHS in 
eq.(\ref{Ricattiequation}).
We first consider 
the following quantity which often appears in the computation of 
$J(k ; j ; H)$. 
\beq
K(k, l ; j ; H)
&:=&
\int_0^{nt}
\sin (\Delta\theta_s)
e^{2ik \theta_s(\kappa_c) + 2il \theta_s(\kappa_d)}
a(s)^j 
H(X_s) ds, 
\\
&&
k, l \in {\bf Z}, 
\quad
j \in {\bf N}, 
\quad
H \in C^{\infty}(M).
\eeq
\begin{lemma}
\label{recursion}
Suppose 
$(k, l) \ne (0,0)$
and 
$j \ge 2$.
Then
\beq
&&
K(k, l ; j ; H)
\\
& \approx &
- \frac {2ik}{2 \kappa_0}
\Biggl\{
\frac 12 \cdot 
K(k+2, l ; j+1 ; F R_{2 (k+l) \kappa_0} H)
+
\frac 12 \cdot
K(k-2, l ; j+1; F R_{2 (k+l) \kappa_0} H)
\\
&& \qquad \qquad \qquad -  
K(k, l ; j+1; F R_{2 (k+l) \kappa_0} H)
\Biggr\}
\\
&&
- \frac {2il}{2 \kappa_0}
\Biggl\{
\frac 12 \cdot 
K(k, l+2 ; j+1; F R_{2 (k+l) \kappa_0} H)
+
\frac 12 \cdot 
K(k, l-2 ; j+1; F R_{2 (k+l) \kappa_0} H)
\\
&& \qquad \qquad \qquad  - 
K(k, l ; j+1; F R_{2 (k+l) \kappa_0} H)
\Biggr\}.
\eeq
\end{lemma}
\begin{proof}
We note 
\beq
&&
2ik \theta_s(\kappa_c) + 2il \theta_s(\kappa_d)
\\
&=&
2i (k+l) \kappa_0 s
+
\frac {2i (ck+dl)}{n} s
+
2ik \tilde{\theta}_s(\kappa_c)
+
2il \tilde{\theta}_s (\kappa_d).
\eeq
Using 
(\ref{Itoformula}) 
with 
$\kappa = 2 (k+l)\kappa_0$ 
and  
$f = H$, 
we have 
\beq
&&
K(k, l ; j ; H)
\\
&=&
\int_0^{nt}
\sin (\Delta\theta_s)
\exp 
\left[
\frac {2i(ck+dl)}{n} s 
+
2i k \tilde{\theta}_s(\kappa_c)
+
2i l \tilde{\theta}_s(\kappa_d)
\right]
a(s)^j 
e^{2i(k+l) \kappa_0 s} H(X_s) ds
\\
&=&
\int_0^{nt}
\sin (\Delta\theta_s)
\exp 
\left[
\frac {2i(ck+dl)}{n} s 
+
2i k \tilde{\theta}_s(\kappa_c)
+
2i l \tilde{\theta}_s(\kappa_d)
\right]
a(s)^j 
\\
&& \qquad\qquad \times 
\left\{
d \left(
e^{2i (k+l) \kappa_0 s}
R_{2(k+l) \kappa_0} H(X_s)
\right)
-
e^{2i (k+l) \kappa_0 s}
\nabla R_{2(k+l) \kappa_0} H(X_s)
d X_s
\right\}
\eeq
For simplicity, we set 
\beq
\widetilde{H} := R_{2(k+l) \kappa_0} H.
\eeq
Integration by parts yields 
\beq
&&
K(k, l ; j ; H)
\\
&=&
\left[
\sin (\Delta \theta_s)
e^{2ik \theta_s (\kappa_c) + 2i \theta_s(\kappa_d)}
a(s)^j
\widetilde{H}(X_s)
\right]_0^{nt}
\\
&& - 
\int_0^{nt} 
\cos (\Delta \theta_s)
\\
&& \times 
\Biggl\{
\frac {c-d}{n}
+
\frac {1}{2 \kappa_c}
Re \left[
e^{2i \theta_s(\kappa_c)} - 1
\right]
a(s) F(X_s)
-
\frac {1}{2 \kappa_d}
Re \left[
e^{2i \theta_s(\kappa_d)} - 1
\right]
a(s) F(X_s)
\Biggr\}
\\
&& \times 
e^{2i k \theta_s(\kappa_c) + 2il \theta_s(\kappa_d)}
a(s)^j 
\widetilde{H}(X_s) ds
\\
&& -
\int_0^{nt}
\sin (\Delta \theta_s)
\\
&& \times
\Biggl\{
\frac {2i(ck+dl)}{n}
+
\frac {2ik}{2 \kappa_c}
Re \left(
e^{2i \theta_s(\kappa_c)} - 1
\right)
a(s) F(X_s)
\\
&& \qquad
+
\frac {2il}{2 \kappa_d}
Re \left(
e^{2i \theta_s(\kappa_d)} - 1
\right)
a(s) F(X_s)
\Biggr\}
\\
&& \times 
e^{2i k \theta_s(\kappa_c) + 2il \theta_s(\kappa_d)}
a(s)^j 
\widetilde{H}(X_s) ds
\\
&& - 
\int_0^{nt}
\sin (\Delta \theta_s)
e^{2i k \theta_s(\kappa_c) + 2il \theta_s(\kappa_d)}
(a(s)^j)' 
\widetilde{H}(X_s) ds
\\
&& - 
\int_0^{nt}
\sin (\Delta \theta_s)
e^{2i k \theta_s(\kappa_c) + 2il \theta_s(\kappa_d)}
a(s)^j 
\nabla \widetilde{H}(X_s) dX_s
\\
&=:&
K_1 + \cdots + K_5.
\eeq
Then 
$K_1 = O(n^{-\alpha}) \approx 0$.
Since 
$j \ge 2$, 
$K_2$, $K_4$ 
is included in 
$E(nt)$ 
and thus negligible : 
$K_2, K_4 \approx 0$.
$K_5$ 
is a martingalge with negligible quadratic variation : 
$\langle K_5 \rangle = O\left(
\int_0^{nt} a(s)^{2j} ds \right)$
so that 
$K_5 \approx 0$. 
Therefore 
\beq
K(j ; k,l ; H)
\approx
K_3.
\eeq
In the integrand of 
$K_3$, 
the 1st term has 
$O(n^{-1})$ 
factor and thus negligible.
In the 2nd and 3rd terms, 
we can replace 
$2ik/2 \kappa_c$, $2il/2 \kappa_d$
by 
$2ik/2 \kappa_0$, $2il/2 \kappa_0$
respectively which produces negligible 
$O(n^{-1})$
error.
Hence 
\beq
&&
K_3
\\
&\approx&
-
\int_0^{nt}
\sin (\Delta \theta_s)
\\
&& \times
\Biggl\{
\frac {2ik}{2 \kappa_0}
Re \left(
e^{2i \theta_s(\kappa_c)} - 1
\right)
a(s) F(X_s)
+
\frac {2il}{2 \kappa_0}
Re \left(
e^{2i \theta_s(\kappa_d)} - 1
\right)
a(s) F(X_s)
\Biggr\}
\\
&& \times 
e^{2i k \theta_s(\kappa_c) + 2il \theta_s(\kappa_d)}
a(s)^j 
\widetilde{H}(X_s) ds
\\
&=&
- \frac {2ik}{2 \kappa_0}
\Biggl\{
\frac 12 \cdot 
K(k+2, l ; j+1 ; F \cdot\widetilde{H})
+
\frac 12
K(k-2, l ; j+1 ; F \cdot\widetilde{H})
\\
&& \qquad \qquad \qquad -  
K(k, l ; j+1 ; F \cdot\widetilde{H})
\Biggr\}
\\
&&
- \frac {2il}{2 \kappa_0}
\Biggl\{
\frac 12 \cdot 
K(k, l+2 ; j+1 ; F \cdot\widetilde{H})
+
\frac 12 \cdot 
K(k, l-2 ; j+1 ; F \cdot\widetilde{H})
\\
&& \qquad \qquad \qquad  - 
K(k, l ; j+1 ; F \cdot\widetilde{H})
\Biggr\}.
\eeq
\QED
\end{proof}
\begin{lemma}
\label{cancel}
Suppose 
$(k, l) \ne (0,0)$
and 
$j \ge 2$.
Then 
\beq
K(k, l ; j ; H) - K(l, k ; j ; H)
\approx 
0.
\eeq
\end{lemma}
\begin{proof}
We compute 
each terms by Lemma \ref{recursion}.
If we have 
terms of the form 
$K(0, 0 ; j+1 ; H')$, 
it equally comes from the 1st and 2nd terms and cancels each other. 
Therefore 
the terms of the form
$K(k', l' ; j+1 ; H')$ 
with 
$(k', l') \ne (0,0)$ 
only remain so that we can continue to use 
Lemma \ref{recursion} 
at least for 
$(j_0 - j)$ - times 
so that the quantity in question is equal to the sum of the terms of the form 
$K(k', l' ; j_0 ; H')$ 
which are negligible.
\QED
\end{proof}
Here we recall 
the definition of 
$J(k ; j ; H)$ : 
\beq
J(k ; j ; H)(\kappa)
&:=&
\int_0^{nt}
\frac {1}{\sin (\Delta \theta_s)}
e^{2i k \theta_s(\kappa)}
a(s)^j 
H(X_s) ds
\eeq
where 
$k \in {\bf Z}$,
$j \ge 1$, 
and 
$H \in C^{\infty}(M)$.
We compute 
$J(k ; j ; H)$ 
by using 
Lemmas \ref{recursion}, \ref{cancel}. 
\begin{proposition}
\label{J}
\mbox{}\\
(1)
$j=1$, 
$k=1$ : 
\begin{eqnarray}
&&
\Delta J(k ; j ; H)
\nonumber
\\
&\approx&
\frac {1}{\kappa_0}
\langle F\cdot R_{2k \kappa_0}H \rangle
\int_0^{nt}
\cos(\Delta \theta_s)
a(s)^{j+1} 
ds
\nonumber
\\
&&
- 
\frac {2ik}{2 \kappa_0}
\Biggl\{
\frac 12
\Delta J(k+1 ; j+1 ; F R_{2k \kappa_0}H)
-
\Delta J(k ; j+1 ; F R_{2k \kappa_0}H)
\Biggr\}
\nonumber
\\
&&+ M_t
\label{computeJ}
\end{eqnarray}
where 
$M$ 
is a martingale whose quadratic variation satisfies 
\beq
\langle M, M \rangle_t
&=&
o \left(
\int_0^{nt} a(s)^{2j} ds
\right)
\\
\langle M, \overline{M} \rangle_t
&=&
4 \langle \psi \rangle
\int_0^{nt} a(s)^{2j} ds
(1 + o(1)), 
\quad
\psi := 
[
R_{2k \kappa_0}(H), 
\overline{R_{2k \kappa_0}(H)}
].
\eeq
(2)
$j \ge 2$, 
$k \ne 0$ : 
\begin{eqnarray}
&&
\Delta J(k ; j ; H)
\nonumber
\\
&\approx&
- 
\frac {2ik}{2 \kappa_0}
\Biggl\{
\frac 12
\Delta J(k+1 ; j+1 ; F R_{2k \kappa_0}H)
+
\frac 12
\Delta J(k-1 ; j+1 ; F R_{2k \kappa_0}H)
\nonumber
\\
&& \qquad\qquad
-
\Delta J(k ; j+1 ; F R_{2k \kappa_0}H)
\Biggr\}.
\label{computeJ2}
\end{eqnarray}
\end{proposition}
\begin{proof}
We use (\ref{Itoformula}) with 
$k = 2k \kappa_0$.
Setting 
$\widetilde{H} := R_{2k \kappa_0} H$
for simplicity, we have 
\beq
&&
J(k ; j ; H) (\kappa_x)
\\
&=&
\left[
\frac {1}{\sin (\Delta \theta_s)}
e^{2ik \theta_s (\kappa_x)}
a(s)^j
\widetilde{H}(X_s)
\right]_0^{nt}
\\
&&+
\int_0^{nt}
\frac 
{\cos (\Delta \theta_s)}
{\sin^2 (\Delta \theta_s)}
\Biggl\{
\frac {c-d}{n} 
+
\frac {1}{2 \kappa_c} Re 
\left(
e^{2i \theta_s(\kappa_c)} - 1 
\right)
a(s) F(X_s)
\\
&& \qquad
-
\frac {1}{2 \kappa_d} Re 
\left(
e^{2i \theta_s(\kappa_d)} - 1 
\right)
a(s) F(X_s)
\Biggr\}
e^{2ik \theta_s(\kappa_x)}
a(s)^j 
\widetilde{H}(X_s) ds
\\
&& - 
\int_0^{nt}
\frac {1}{\sin (\Delta \theta_s)}
\left\{
2ik \cdot \frac xn
+
\frac {2ik}{2 \kappa_x} Re 
\left(
e^{2i \theta_s(\kappa_x)}-1
\right)
a(s) F(X_s)
\right\}
\\
&& \qquad\times
e^{2ik \theta_s(\kappa_x)}
a(s)^j 
\widetilde{H} (X_s) ds
\\
&& - 
\int_0^{nt}
\frac {1}{\sin (\Delta \theta_s)}
e^{2ik \theta_s(\kappa_x)}
(a(s)^j)'
\widetilde{H} (X_s) ds
\\
&& - 
\int_0^{nt}
\frac {1}{\sin (\Delta \theta_s)}
e^{2ik \theta_s(\kappa_x)}
a(s)^j
\nabla \widetilde{H} (X_s) dX_s
\\
&=:&
J_1 + \cdots + J_5.
\eeq
We estimate 
$\Delta J_1, \cdots, \Delta J_5$ 
separately. 
It will turn out that 
$\Delta J_1$, $\Delta J_4$ 
are negligible, 
$\Delta J_2$ 
is equal to the 1st term of RHS in 
(\ref{computeJ}) 
modulo error, 
$\Delta J_3$ 
is equal to the 2nd term of RHS in 
(\ref{computeJ}) 
or 
is equal to RHS in 
(\ref{computeJ2}). \\
\noindent
(1)
$J_1$ : 
By an elementary equality 
\begin{equation}
e^{2i \theta_1} - e^{2i \theta_2}
=
2i \sin (\theta_1 - \theta_2) 
e^{i  \theta_1 + i \theta_2}
\label{sin}
\end{equation}
we have 
\beq
\Delta J_1
&=&
\left[
\frac {1}{\sin (\Delta \theta_s)}
\cdot
2i \sin (k \Delta \theta_s)
e^{ik (\theta_s (\kappa_c) + \theta_s(\kappa_d))}
a(s)^j
\widetilde{H} (X_s)
\right]_0^{nt}.
\eeq
Therefore 
$\Delta J_1 = O(n^{-j\alpha})+C \approx 0$. \\
(2)
$J_2$ : we separate the discussion into the following two cases.\\
(i)
$j \ge 2$ : 
as in the proof of Lemma \ref{recursion}, 
we may ignore the term with 
$(c-d)/n$ 
factor and replace 
$1 / 2 \kappa_c$, $1 / 2 \kappa_c$ 
by 
$1 / 2 \kappa_0$ : 
\beq
&&
J_2
\\
&\approx&
\int_0^{nt}
\frac 
{\cos (\Delta \theta_s)}
{\sin^2 (\Delta \theta_s)}
\frac {1}{2 \kappa_0} Re 
\left(
e^{2i \theta_s(\kappa_c)} - e^{2i \theta_s(\kappa_d)}
\right)
e^{2ik \theta_s(\kappa_x)}
a(s)^{j+1}
(F\cdot \widetilde{H})(X_s) ds. 
\eeq
And we compute 
$\Delta J_2$ 
using 
(\ref{sin}) : 
\beq
&&
\Delta J_2
\\
&\approx&
\int_0^{nt}
\frac 
{\cos (\Delta \theta_s)}
{\sin^2 (\Delta \theta_s)}
\frac {1}{2 \kappa_0} Re 
\left(
e^{2i \theta_s(\kappa_c)} - e^{2i \theta_s(\kappa_d)}
\right)
\left(
e^{2ik \theta_s(\kappa_c)}
-
e^{2ik \theta_s(\kappa_x)}
\right)
\\
&&
a(s)^{j+1}
(F\cdot \widetilde{H})(X_s) ds
\\
&=&
\int_0^{nt}
\frac 
{\cos (\Delta \theta_s)}
{\sin (\Delta \theta_s)}
\cdot
\sin (k \Delta \theta_s)
\frac {1}{2 \kappa_0} Re 
\left[
2i 
e^{i (\theta_s (\kappa_c) + \theta_s(\kappa_d))}
\right]
\left(
2i 
e^{ik (\theta_s (\kappa_c) + \theta_s(\kappa_d))}
\right)
\\
&& \qquad \times
a(s)^{j+1}
(F\cdot \widetilde{H})(X_s) ds
\eeq
which is negligible if 
$j \ge 2$ : 
$\Delta J_2 \approx 0$.
\\
(ii)
$j=1$, $k=1$ : 
we further decompose as follows. 
\beq
\Delta J_2
&\approx&
\frac {1}{\kappa_0}
\int_0^{nt}
\cos (\Delta \theta_s)
\left(
1-
e^{2i (\theta_s (\kappa_c) + \theta_s(\kappa_d))}
\right)
a(s)^{j+1}
(F\cdot \widetilde{H})(X_s) ds
\\
&=:& 
\Delta J_{2-1} + \Delta J_{2-2}.
\eeq
For 
$\Delta J_{2-1}$, 
we use (\ref{SecondItoformula}) : 
\beq
&&
\Delta J_{2-1}
\\
&=&
\frac {1}{\kappa_0}
\int_0^{nt}
\cos (\Delta \theta_s)
a(s)^{j+1}
\left\{
\langle F\cdot \widetilde{H} \rangle 
-
d\left(
R_0 (F\cdot \widetilde{H})
\right)
-
\nabla R_0 (F\cdot \widetilde{H})
d X_s
\right\}
\\
&=:&
\Delta J_{2-1-1} + \cdots + \Delta J_{2-1-3}.
\eeq
$\Delta J_{2-1-1}$ 
already has the desired form. 
For 
$\Delta J_{2-1-2}$, 
integration by parts yields 
\beq
&&
\Delta J_{2-1-2}
\\
&=&
\frac {1}{\kappa_0}
\Biggl\{
\left[
\cos (\Delta \theta_s)
a(s)^{j+1}
R_0 (F\cdot \widetilde{H})(X_s)
\right]_0^{nt}
\\
&& + 
\int_0^{nt}
\sin (\Delta \theta_s)
\Biggl\{
\frac {c-d}{n}
+
\frac  {1}{2 \kappa_c} Re 
\left(
e^{2i \theta_s(\kappa_c)}-1
\right)
a(s) F(X_s)
\\
&& \qquad
-
\frac {1}{2 \kappa_d} Re 
\left(
e^{2i \theta_s(\kappa_d)}-1
\right)
a(s) F(X_s)
\Biggr\}
a(s)^{j+1}
R_0 (F\cdot \widetilde{H})(X_s) ds
\\
&& - 
\int_0^{nt}
\cos (\Delta \theta_s)
( a(s)^{j+1} )'
R_0 (F\cdot \widetilde{H})(X_s)ds.
\eeq
As in the proof of Lemma \ref{recursion}, 
1st and 3rd terms are negligible ; in the 2nd term, the term with 
$(c-d)/n$ 
factor is also negligible and 
$1 / 2 \kappa_c$, $1 / 2 \kappa_d$ 
may be replaced by 
$1 / 2 \kappa_0$ 
up to negligible error : 
\beq
&&
\Delta J_{2-1-2} 
\\
& \approx &
\frac {1}{\kappa_0}
\int_0^{nt}
\sin (\Delta \theta_s)
\frac {1}{2 \kappa_0}Re 
\left(
e^{2i \theta_s(\kappa_c)}
-
e^{2i \theta_s(\kappa_d)}
\right)
a(s)^{j+2}
F \cdot R_0 (F\cdot \widetilde{H})(X_s)ds
\\
&=&
\frac {1}{\kappa_0}
\cdot
\frac {1}{2 \kappa_0}
\\
&& 
\times
Re 
\Biggl\{
K(2, 0 \,;\, j+2 \,;\, F \cdot R_0 (F\cdot \widetilde{H}))
-
K(0, 2 \,;\, j+2 \,;\, F \cdot R_0 (F\cdot \widetilde{H}))
\Biggr\}
\\
& \approx & 0.
\eeq
In the last step, we used Lemma \ref{cancel}.
For 
$\Delta J_{2-1-3}$, 
\beq
\langle \Delta J_{2-1-3}, \Delta J_{2-1-3} \rangle
= O \left(
\int_0^{nt} a(s)^{2j+2} 
\right)
=
o \left(
\int_0^{nt} a(s)^{2j} 
\right)
\eeq
so that 
$\Delta J_{2-1-3} \approx 0$.
Therefore, we have 
\beq
\Delta J_{2-1}
\approx 
\Delta J_{2-1-1}
=
\frac {1}{\kappa_0}
\langle F\cdot \widetilde{H} \rangle
\int_0^{nt}
\cos(\Delta \theta_s)
a(s)^{j+1} 
ds.
\eeq
For 
$\Delta J_{2-2}$, 
we use 
(\ref{Itoformula}) 
with 
$\kappa = 4 \kappa_0$, 
perform the integration by parts, estimate as before, and use Lemma \ref{cancel} : 
\beq
&&
\Delta J_{2-2}
\\
&=&
-
\frac {1}{\kappa_0}
\Biggl\{
\left[
\cos (\Delta \theta_s)
e^{2i (\theta_s(\kappa_c)+\theta_s(\kappa_d))}
a(s)^{j+1}
R_{4 \kappa_0}(F \cdot \widetilde{H})(X_s)
\right]_0^{nt}
\\
&&\qquad +
\int_0^{nt}
\sin (\Delta \theta_s)
\Biggl\{
\frac {c-d}{n}
+
\frac {1}{2 \kappa_c} Re 
\left(
e^{2i \theta_s(\kappa_c)}-1
\right)
a(s) F(X_s)
\\
&& \qquad\qquad
-
\frac {1}{2 \kappa_d} Re 
\left(
e^{2i \theta_s(\kappa_d)}-1
\right)
a(s) F(X_s)
\Biggr\}
e^{2i (\theta_s(\kappa_c)+\theta_s(\kappa_d))}
\\
&&\qquad\qquad \times
a(s)^{j+1}
R_{4 \kappa_0}(F \cdot \widetilde{H})(X_s) ds
\\
&&\qquad -
\int_0^{nt}
\cos (\Delta \theta_s)
\Biggl\{
2i \cdot
\frac {c+d}{n}
+
\frac {1}{2 \kappa_c} Re 
\left(
e^{2i \theta_s(\kappa_c)}-1
\right)
a(s) F(X_s)
\\
&& \qquad\qquad
+
\frac {1}{2 \kappa_d} Re 
\left(
e^{2i \theta_s(\kappa_d)}-1
\right)
a(s) F(X_s)
\Biggr\}
e^{2i (\theta_s(\kappa_c)+\theta_s(\kappa_d))}
\\
&&\qquad\qquad \times
a(s)^{j+1}
R_{4 \kappa_0}(F \cdot \widetilde{H})(X_s) ds
\\
&&\qquad - 
\int_0^{nt}
\cos (\Delta \theta_s)
e^{2i (\theta_s(\kappa_c)+\theta_s(\kappa_d))}
( a(s)^{j+1} )'
R_{4 \kappa_0}(F \cdot \widetilde{H})(X_s) ds
\\
&&\qquad - 
\int_0^{nt}
\cos (\Delta \theta_s)
e^{2i (\theta_s(\kappa_c)+\theta_s(\kappa_d))}
a(s)^{j+1}
\nabla R_{4 \kappa_0}(F \cdot \widetilde{H})(X_s) dX_s
\Biggr\}
\\
& \approx &
- \frac {1}{\kappa_0}
\int_0^{nt}
\sin (\Delta \theta_s)
\frac {1}{2 \kappa_0} Re 
\left(
e^{2i \theta_s(\kappa_c)}
-
e^{2i \theta_s(\kappa_d)}
\right)
e^{2i (\theta_s(\kappa_c)+\theta_s(\kappa_d))}
\\
&&\qquad\qquad \times
a(s)^{j+2}
F \cdot R_{4 \kappa_0}(F \cdot \widetilde{H})(X_s) ds
\\
&=&
-\frac {1}{\kappa_0} \cdot \frac {1}{2 \kappa_0}
\cdot \frac 12
\\
&& 
\times
\Biggl\{
K(4, 2 ; j+2 ; F \cdot R_{4 \kappa_0}(F \cdot \widetilde{H}))
+ 
K(0, 2 ; j+2 ; F \cdot R_{4 \kappa_0}(F \cdot \widetilde{H}))
\\
&& \quad -
K(2,4 ; j+2 ; F \cdot R_{4 \kappa_0}(F \cdot \widetilde{H}))
-
K(2,0 ; j+2 ; F \cdot R_{4 \kappa_0}(F \cdot \widetilde{H}))
\Biggr\}
\\
&\approx& 0.
\eeq
To summarize : 
\beq
\Delta J_{2}
\approx
\frac {1}{\kappa_0}
\langle F\cdot \widetilde{H} \rangle
\int_0^{nt}
\cos(\Delta \theta_s)
a(s)^{j+1} 
ds.
\eeq
\noindent
(3)
$J_3$ : 
after cutting out negligible terms we have 
\beq
&&
J_3
\\
&\approx&
- 
\int_0^{nt}
\frac {1}{\sin (\Delta \theta_s)}
\frac {2ik}{2 \kappa_0} Re 
\left(
e^{2i \theta_s(\kappa_x)}-1
\right)
e^{2ik \theta_s(\kappa_x)}
a(s)^{j+1} 
(F \cdot \widetilde{H})(X_s) ds
\\
&=&
- 
\frac {2ik}{2 \kappa_0}
\int_0^{nt}
\frac {1}{\sin (\Delta \theta_s)}
\left(
\frac 
{
e^{2i(k+1) \theta_s(\kappa_x)}
+
e^{2i(k-1) \theta_s(\kappa_x)}
}
{2}
-
e^{2ik \theta_s(\kappa_x)}
\right)
\\
&& \qquad\times 
a(s)^{j+1} 
(F \cdot \widetilde{H})(X_s) ds
\\
&=&
- 
\frac {2ik}{2 \kappa_0}
\Biggl\{
\frac 12
J(k+1 ; j+1 ; F \widetilde{H})(\kappa_x)
+
\frac 12
J(k-1 ; j+1 ; F \widetilde{H})(\kappa_x)
\\
&& \qquad
-
J(k ; j+1 ; F \widetilde{H})(\kappa_x)
\Biggr\}. 
\eeq
(4)
$J_4$ : this is clearly negligible : 
\beq
\Delta J_4
&=& -
\int_0^{nt}
\frac {1}{\sin (\Delta \theta_s)}
\cdot
2i \sin (k \Delta \theta_s)
e^{ik (\theta_s(\kappa_c) + \theta_s(\kappa_d))}
(a(s)^j)'
\widetilde{H}(X_s) ds
\approx 0.
\eeq
(5)
$J_5$ : using (\ref{sin}) we have 
\beq
\Delta J_5
&=&
- \int_0^{nt}
\frac {1}{\sin (\Delta \theta_s)}
2i \sin (k \Delta \theta_s)
e^{ik (\theta_s(\kappa_c) + \theta_s(\kappa_d))}
a(s)^j
\nabla \widetilde{H}(X_s) d X_s
\eeq
We consider the following two cases.
\\
(i)
$k = 1$ : setting 
\beq
\varphi := [ \widetilde{H}, \widetilde{H} ],
\quad
\psi := [
\widetilde{H}, 
\overline{\widetilde{H}}
], 
\eeq
we have 
\beq
\langle \Delta J_5, \Delta J_5 \rangle
&=&
(-4)
\int_0^{nt}
e^{2i (\theta_s(\kappa_c) + \theta_s(\kappa_d))}
a(s)^{2j}
\varphi (X_s) ds
\\
\langle \Delta J_5, \overline{\Delta J_5} \rangle
&=&
4 \int_0^{nt}
a(s)^{2j}
\psi(X_s) ds.
\eeq
to which we apply 
(\ref{Itoformula}), (\ref{SecondItoformula}) 
respectively.
By the same argument 
as in the estimate of 
$\Delta J_2, \Delta J_3$ 
we have 
\beq
\langle \Delta J_5, \Delta J_5 \rangle
&=&
o \left(
\int_0^{nt} a(s)^{2j} ds
\right)
\\
\langle \Delta J_5, \overline{\Delta J_5} \rangle
&=&
4 \langle \psi \rangle
\int_0^{nt} a(s)^{2j} ds
(1 + o(1)), 
\quad
n \to \infty.
\eeq
(ii)
$k \ge 2$ : 
by a direct computation, it is easy to see 
\beq
\langle 
\Delta J_5, \Delta J_5 
\rangle, 
\;
\langle 
\Delta J_5, \overline{\Delta J_5}
\rangle
&=&
O \left(
\int_0^{nt}
a(s)^{2j} ds
\right)
\eeq
so that
$\Delta J_5 \approx 0$
for
$j \ge 2$. 
\QED
\end{proof}
%

\section{Appendix II}
In Appendix II, 
we provide the proofs of Proposition \ref{exponential} and statements in Section 5 for the sake of completeness, all of which are done by tracing those in \cite{AD}.\\
\noindent
{\it Proof of Proposition \ref{exponential}}\\
We discuss 
the computation of  
$t_n^{(+)}(r)$ 
only, for 
$t_n^{(-)}(r)$ 
can be treated similarly. 
We write eq.(\ref{S}) as in the following manner : 
\beq
dS_+
&=&
-
W_+(S_+) dt
+
C_n d W_t
\\
\mbox{where }
\;
-W_+(r)
&:=&
\lambda 
(1 + \epsilon)
\cosh_+ r 
+
\frac {C_n^2}{2}
\tanh_{+, \epsilon} r.
\eeq
Then 
\beq
- V_+(r)
&:=&
\lambda 
(1 + \epsilon)
\left\{
\sinh(r \pm \delta) 
\mp \sinh \delta 
\right\}
1( \pm r > 0)
\\
&& + \frac {C_n^2}{2}
(1 \pm \epsilon)
\log 
\frac {\cosh (r+\delta)}{\cosh \delta}
1(\pm r > - \delta)
\eeq
satisfies 
$V'_+(r) = W_+ (r)$
for 
$r \ne 0, - \delta$.
We 
first derive the critical points 
$r = a_n$, $b_n$ 
such that 
$W_+(r) = 0$ : 
\beq
a_{n}
&=&
\delta + 
\log \frac {\tilde{\lambda}}{C_n^2}
+
O(C_n^{-2})
\\
b_{n}
&=&
- 
\frac {2\tilde{\lambda}}{C_n^2}
\cosh(2 \delta) (1 + O(C_n^{-2})) - \delta
\eeq
where 
$\tilde{\lambda}
:=
(1+\epsilon)\lambda/(1-\epsilon)$.
Moreover we have 
\beq
&&
V_+(a_{n}+y)
\\
&=&-
\lambda
(1 + \epsilon)
\left\{
\frac {\tilde{\lambda}}{2C_n^2}
e^{y +\delta \pm \delta + O(C^{-2})} 
-
\frac {C_n^2}{2\tilde{\lambda}}
e^{-y+\delta \pm \delta  + O(C_n^{-2})}
\mp
\sinh \delta
\right\}
\\
&& \qquad \times 
1(\pm (a_n+x) > 0)
\\
&& - 
\frac {C_n^2}{2}
(1 \pm \epsilon)
\left\{
\log 
\left(
\frac {\tilde{\lambda}}{2C_n^2}
e^{y + 2 \delta  + O(C_n^{-2})}
+
\frac {C_n^2}{2\tilde{\lambda}}
e^{- y -2 \delta +  O(C_n^{-2}))}
\right)
-
\log \cosh \delta
\right\}
\\
&& \qquad \times
1(\pm (a_n+x) > - \delta)
\eeq
\beq
&&
V_+(b_{n}+x)
\\
&=&
- \lambda(1 + \epsilon)
\left\{
\sinh 
\left(
x - \delta \pm \delta 
- \frac {2\tilde{\lambda}}{C_n^2}
\cosh(2 \delta)(1+O(C_n^{-2}))
\right)
\mp
\sinh \delta
\right\}
\\
&& \qquad\times 
1(\pm(b_n+x) > 0)
\\
&& - \frac {C_n^2}{2} (1 \pm \epsilon)
\left\{
\log 
\cosh \left(
x - 
\frac {2\tilde{\lambda}}{C_n^2}
\cosh (2 \delta)
(1 + O(C_n^{-2}))
\right)
-
\log \cosh \delta
\right\}
\\
&& \qquad\times 
1(\pm(b_n+x) > - \delta).
\eeq
Since 
$t_n^{(+)}(r)$ 
satisfies 
\[
\frac {C_n^2}{2} f'' - W_+(r) f' = -1, 
\quad
f(\infty)=0, 
\]
we have
\beq
t_{n}^{(+)}(r)
&=&
\frac {2}{C_n^2} 
\int_r^{\infty} dx \int_{- \infty}^x dy
\;
\exp\left\{
\frac {2}{C_n^2}
\left(
V_+(x) - V_+(y)
\right)
\right\}.
\eeq
Substituting 
above equations, we have 
\beq
&&
t_{n}^+(r)
\\
&=&
\frac {2}{C_n^2}
\int_{r-b}^{\infty} dx 
\\
&& 
\exp
\Biggl[
- \frac {2}{C_n^2}
\lambda(1 + \epsilon)
\left\{
\sinh 
\left(
x - \delta \pm \delta 
- 
\frac {2 \tilde{\lambda}}{C_n^2}
\cosh (2 \delta) (1 + O(C_n^{-2}))
\right)
\mp
\sinh \delta 
\right\}
\\
&& \qquad \times
1(\pm(b_n+x) > 0)
\Biggr]
\\
&& \cdot
\left\{
\frac {\cosh \delta}
{
\cosh 
\left(
x - \frac {2 \tilde{\lambda}}{C_n^2} \cosh (2 \delta) 
(1 + O(C_n^{-2}))
\right)
}
\right\}^{1 \pm \epsilon}
1(\pm(b_n+x) > - \delta)
\\
&& \times 
\int_{- \infty}^{(b-a)+x} dy
\\
&&
\exp
\Biggl[
\left\{
\frac {\lambda \cdot \tilde{\lambda} (1 + \epsilon)}{C_n^4}
e^{y + \delta \pm \delta + O(C_n^{-2})} 
-
\frac {\lambda(1 + \epsilon)}{\tilde{\lambda}}
e^{-y + \delta \pm \delta + O(C_n^{-2})}
\mp
\frac {2}{C_n^2} \cdot \lambda (1 + \epsilon) \cdot \sinh \delta 
\right\}
\\
&& \qquad \times
1(\pm(a_n+x) > 0)
\Biggr]
\\
&& \cdot
\left(
\frac {\tilde{\lambda}}{2 C_n^2}
e^{y + 2 \delta + O(C_n^{-2})}
+
\frac {C_n^2}{2 \tilde{\lambda}}
e^{-y - 2 \delta + O(C_n^{-2})}
\right)^{1 \pm \epsilon}
\cdot
\frac {1}{(\cosh \delta)^{1 \pm \epsilon}}
\cdot
1( \pm (a_n+x) > - \delta)
\eeq
Noting that 
$\epsilon \to 0$, 
$\tilde{\lambda} \to \lambda$, 
$a_{n} \to - \infty$, 
$b_{n} \to 0$
as 
$n \to \infty$, 
we have 
\beq
t_n^{(+)}(r)
\stackrel{n \to \infty}{\to}
\frac {1}{\lambda}
\int_r^{\infty} 
\frac {dx}{\cosh x}
\int_{- \infty}^{\infty} dy 
\;
e^{- y -  e^{-y} }
=
\frac {1}{\lambda}
\int_r^{\infty} 
\frac {dx}{\cosh x}.
\eeq
Thus
\beq
\lim_{r \downarrow - \infty}
\lim_{n \to \infty}
t_n^{(+)}(r)
=
\frac {\pi}{\lambda}.
\eeq
The statement 
for the Laplace transform is derived by the same way as in the proof of Proposition 2.2 \cite{AD}. 
\QED\\
%

%
\noindent
{\it Proof of Lemma \ref{explosion1}\\
LHS 
of the inequality in question is bounded from below by 
\beq
LHS
& \ge &
{\bf P}_{
\epsilon 
\log n^{\frac {1}{\gamma}}
}
\left(
T_{c \log n^{\frac {1}{\gamma}}}
<
\frac {4c}{C_n^2}
\log n^{\frac {1}{\gamma}}
\wedge
T_{
\frac {\epsilon}{2}
\log n^{\frac {1}{\gamma}}
}
\right)
\\
&& \qquad \times
{\bf P}_{
c \log n^{\frac {1}{\gamma}},
\; 
\frac {4c}{C_n^2}
\log n^{\frac {1}{\gamma}}
}
\left(
T_{+\infty} < 
\frac {c}{C_n^2}
\log n^{\frac {1}{\gamma}}
\right)
\\
&=:& (1) \times (2).
\eeq
which we estimate separately. \\
(1)
If 
$
\frac {\epsilon}{2}
\log n^{\frac {1}{\gamma}}
< r < 
c 
\log n^{\frac {1}{\gamma}}
$, 
the drift term of the SDE for 
$R_-$
satisfies 
$
\mbox{(drift term)}
\ge
\frac 12 C_n^2
\tanh r 
\ge 
\frac 14 C_n^2 
$
so that the first factor 
(1) 
is bounded from below by the probability of the following event. 
\beq
{\cal E}
:=
\left\{
\inf_{
0 < t < 4 \frac {c}{C_n^2}
\log n^{\frac {1}{\gamma}}
}
C_n B_t
\ge
-\frac {\epsilon}{2}
\log n^{\frac {1}{\gamma}}
\right\}
\eeq
where 
$B_t$ 
is a Brownian motion with 
$B_0 = 0$.
By the reflection principle, we have 
\beq
{\bf P} \left( {\cal E} \right)
& = &
{\bf P}
\left(
C_n 
\left|
B 
\left(
\frac {4c}{C_n^2}
\log n^{\frac {1}{\gamma}}
\right)
\right|
\le
\frac {\epsilon}{2} 
\log n^{\frac {1}{\gamma}}
\right)
\ge
1- 
\left(
n^{-\frac {1}{\gamma}}
\right)^{c'}.
\eeq
(2)
Let 
\beq
\tilde{\cal E}
&:=&
\left\{
\sup_{
0 \le t \le 
\frac {c}{C_n^2} 
\log n^{\frac {1}{\gamma}}
}
C_n |B(t)|
<
\frac {\epsilon}{2}
\log n^{\frac {1}{\gamma}}
\right\}.
\eeq
Then 
$
{\bf P}
\left(
\tilde{\cal E}
\right)
\ge
1 - 
\left(
n^{-\frac {1}{\gamma}}
\right)^{c''}
$
for some 
$c''>0$, 
and under the event 
$\tilde{{\cal E}}$, 
$G(t) := R(t) - C_n B(t)$
satisfies 
\beq
G'(t)
&\ge&
\frac {\lambda}{2}
e^{G(t)}
\cdot
e^{
-\frac {\epsilon}{2}
\log n^{\frac {1}{\gamma}}
}
\cdot
\gamma
\cdot
\left(
\frac {4c}{C_n^2}
\log n^{\frac {1}{\gamma}}
\right)^{\gamma-1}
+
\frac {C_n^2}{2} 
\tanh (G(t) + C_n B)
\\
&\ge &
C
\cdot
\left(
n^{-\frac {1}{\gamma}}
\right)^{
\gamma-1 + \frac {\epsilon}{2}
}
e^{G(t)}
-
\frac {C_n^2}{2}.
\eeq
Therefore 
the explosion time of 
$G$ 
satisfies 
$T_{+\infty}
\sim
\left(
n^{-\frac {1}{\gamma}}
\right)^{
c - 
(\gamma - 1 + \frac {\epsilon}{2})
}$.
\QED\\
%

%
\noindent
{\it Proof of Lemma \ref{explosion2}}\\
LHS 
of the inequality in question is bounded from below by 
\beq
LHS
& \ge &
{\bf P}_{
- \frac 14 
\log n^{\frac {1}{\gamma}}
}
\left(
T_{ 
\frac 14 \log n^{\frac {1}{\gamma}}
}
<
\frac {1}{C_n^2}
\log n^{\frac {1}{\gamma}}
\right)
\\
&& \qquad
\times
{\bf P}_{
\frac 14 
\log n^{\frac {1}{\gamma}}, 
\;
\frac {1}{C^2}
\log n^{\frac {1}{\gamma}}
}
\left(
T_{+\infty}
<
5 \frac {c}{C_n^2}
\log n^{\frac {1}{\gamma}}
\right)
\\
&=:&
(1) \times (2).
\eeq
The second factor 
(2) 
has been estimated in Lemma \ref{explosion1}.
For the first factor (1), 
since 
$
R^{(n)}(t) \ge - \frac {C_n^2}{2} t + C_n B_t
$
we have 
\beq
(1)
& \ge &
{\bf P}
\left(
- \frac {C_n^2}{2}
\cdot
\frac {1}{C_n^2}
\log n^{\frac {1}{\gamma}}
+
C_n
B_{
\frac {1}{C_n^2} 
\log n^{\frac {1}{\gamma}}
}
\ge 
\frac 12
\log n^{\frac {1}{\gamma}}
\right)
\ge
C
\left(
n^{-\frac {1}{\gamma}}
\right)^{1/2}.
\eeq
\QED\\
%

%
\noindent
{\it Proof of Lemma \ref{Xi}}\\
Conditioning at time 
$u$ 
and using the Markov property, we have 
\beq
&&
{\bf P}_{- \infty}
\left(
R^{(n)}(u) \ge - \frac 14 
\log n^{\frac {1}{\gamma}}
\right)
\\
& \le &
{\bf P}_{- \infty}
\left(
R^{(n)}(u) \ge - \frac 14 
\log n^{\frac {1}{\gamma}}, 
\;
T_{+ \infty} 
<
\frac {5c+1}{C_n^2}
\log n^{\frac {1}{\gamma}}
\right)
\\
&& \qquad+
{\bf P}_{- \infty}
\left(
R^{(n)}(u) \ge - \frac 14 
\log n^{\frac {1}{\gamma}}, 
\;
T_{+ \infty} 
\ge
\frac {5c+1}{C_n^2}
\log n^{\frac {1}{\gamma}}
\right)
\\
& \le &
{\bf P}_{- \infty}
\left(
R^{(n)}(u) \ge - \frac 14
\log n^{\frac {1}{\gamma}}, 
\;
T_{+\infty} < 
\frac {5c+1}{C^2} 
\log n^{\frac {1}{\gamma}}
\right)
\\
&&\qquad+
\left(
1 - 
\left(
\frac {1}{n^{\frac {1}{\gamma}}}
\right)^{1/2}
\right)
{\bf P}_{- \infty}
\left(
R^{(n)}(u) \ge 
- \frac 14
\log n^{\frac {1}{\gamma}}
\right).
\eeq
Hence
\beq
&&
{\bf P}_{- \infty}
\left(
R^{(n)}(u)
\ge - \frac 14 
\log n^{\frac {1}{\gamma}}
\right)
\\
& \le &
\left( 
n^{\frac {1}{\gamma}} 
\right)^{1/2}
{\bf P}_{- \infty}
\left(
R^{(n)}(u)
\ge 
- \frac 14 
\log n^{\frac {1}{\gamma}}, 
\;
T_{+ \infty}
<
\frac {5c+1}{C_n^2}
\log n^{\frac {1}{\gamma}}
\right)
\\
& \le &
(n^{\frac {1}{\gamma}})^{1/2}
{\bf P}_{- \infty}
\left(
[
u, 
u+ \frac {5c+1}{C_n^2} 
\log n^{\frac {1}{\gamma}}
]
\mbox{ contains at least one explosion }
\right).
\eeq
Let 
$
k :=
\sharp 
\left\{
\mbox{ explosions of $R^{(n)}$ in $[0,t]$ }
\right\}
$
with 
$\{ \zeta_j \}_{j=1}^k$ 
being the explosion points, we have 
\beq
\int_0^t
1 
\left(
\exists i \; : \; 
\zeta_i 
\in 
\left[
u, u + 
\frac {5c+1}{C_n^2} 
\log n^{\frac {1}{\gamma}}
\right]
\right) du
\le
\frac {5c+1}{C_n^2}
\cdot
\log n^{\frac {1}{\gamma}}
\cdot
(k+1).
\eeq
It thus 
suffices to take the expectation of both sides and use the following inequality : 
${\bf E}[\sharp 
\left\{
\mbox{explosions of $R^{(n)}$ in $[0,t]$}
\right\}
]
\le
{\bf E}[ \Theta_{nt}(\lambda) / \pi ] 
\le C \lambda t / \pi$.
\QED\\
%

\noindent
From now on,
for the sake of simplicity, we use the following notation. 
\[
\Theta^{(n)}_{\lambda}(u)
:=
\Theta_{nu^{\gamma}}(\lambda), 
\quad
\Theta^{(n)}_{\lambda, \lambda'}(u)
:=
\Theta^{(n)}_{\lambda'}(u) - \Theta^{(n)}_{\lambda}(u).
\]
%
%
\noindent
{\it Proof of Proposition \ref{marginal}}\\
It suffices to show, 
\beq
(1)&&
\quad
{\bf E}[ \mu_{\lambda}^{(n)} (I) ] 
\to
\frac {\lambda}{\pi}
\int_I \gamma t^{\gamma-1} 1[0,1]dt
\\
(2)&& \quad
{\bf P}\left(
\mu_{\lambda}^{(n)}(I) = 0
\right)
\to
\exp \left(
- \frac {\lambda}{\pi}
\int_I
\gamma t^{\gamma-1} 1[0,1] dt
\right)
\eeq
for the finite union 
$I \subset [0,1]$ 
of disjoint intervals.
Let 
\beq
{\cal C}_k
&:=&
\left\{
\left\{ 
\Theta^{(n)}_{\lambda}
\left(\frac {T_k}{N}\right) \right\}
\le 
2 \arctan 
\left( n^{-\frac {1}{\gamma}} \right)^{1/4}
\right\}, 
\quad
{\cal C}
:=
\bigcap_{k=1}^{2N+1}
{\cal C}_k.
\eeq
Then by Lemma \ref{indicator}, 
\beq
{\bf P}\left( {\cal C}^c \right)
&\le&
\sum_{k=0}^{2N+1}
{\bf P}\left(
\left\{
\Theta^{(n)}_{\lambda}
\left(\frac {T_k}{N}\right)
\right\}
>
2 \arctan \left(
n^{-\frac {1}{\gamma}}
\right)^{1/4}
\right)
\\
& \le &
2 {\bf E}
\left[
\int_0^{3N+3}
1\left(
\Theta^{(n)}_{\lambda}
\left(\frac {u}{N}\right) > 
2 \arctan \left(
n^{-\frac {1}{\gamma}}
\right)^{1/4}
\right)
du
\right]
\\
& \le &
\left(
n^{-\frac {1}{\gamma}}
\right)^{1/2}
N
\log n^{\frac {1}{\gamma}}
\to 0, 
\quad
n \to \infty.
\eeq
(1)
We may take 
$I = [0,t]$. 
Upper bound simply follows from 
\beq
{\bf E}[ 
\mu_{\lambda}^{(n)}[0,t]
]
&=&
{\bf E}
\left[
\frac {\Theta^{(n)}_{\lambda}(t)}{\pi}
\right]
\le 
\frac {1}{\pi}
\lambda
\int_0^t 
\gamma s^{\gamma-1} ds
=
\frac{\lambda t^{\gamma}}{\pi}.
\eeq
For the lower bound, we consider 
\beq
N_k^{\pm}
&:=&
\sharp \left\{
\mbox{ jumps of
$\Theta^{(n)}_{\lambda, \pm}$ 
in $I_k$ }
\right\}
\\
N_k
&:=&
\sharp \left\{
\mbox{ jumps of
$\Theta^{(n)}_{\lambda}$ 
in $I_k$ }
\right\}.
\eeq
Then
\beq
{\bf E}\left[
\mu_{\lambda}^{(n)}[0,t]
\right]
& \ge &
\sum_{k=0}^{2Nt+1}
{\bf E}\left[
N_k^-
1 \left( \frac {T_k}{N} < t \right)
\middle| {\cal C}_k
\right]
-
\sum_k
{\bf E}\left[
N_k^-
1 \left( \frac {T_k}{N} < t \right)
\middle| {\cal C}_k
\right]
{\bf P}({\cal C}_k^c), 
\eeq
the 2nd term of which vanishes as 
$n \to \infty$ : 
\beq
\mbox{2nd term}
&\le&
{\bf E}
\left[
\sharp 
\left\{
\mbox{ 
jumps of 
$\Theta^{(n)}_{\lambda}$ 
in 
$[0, 3t]$ 
}
\right\}
\right]
\times
\sup_k
\frac {
{\bf P}({\cal C}_k^c)
}
{
{\bf P}({\cal C}_k)
}
\\
& \le &
{\bf E}
\left[
\sharp 
\left\{
\mbox{ 
jumps of 
$\Theta^{(n)}_{\lambda}$ 
in 
$[0, 3t]$ 
}
\right\}
\right]
\times
\frac {
{\bf P}({\cal C}^c)
}
{
1-{\bf P}({\cal C}^c)
}
\to 0.
\eeq
For the 1st term, 
we note that 
${\bf E}[ N_k^- | {\cal C}_k]
=
\pi^{-1}
\lambda
\gamma
\left(
\frac {T_k}{N}
\right)^{\gamma-1}
\cdot
\frac {\tau_{k+1}}{N}$ 
by 
Proposition \ref{exponential}.
Hence 
by the convergence of the Riemannian sum to the integral, 
\beq
\sum_k
{\bf E}[ N_k^- | {\cal C}_k]
&=&
\frac {\lambda}{\pi}
\sum_k
\gamma
\left(
\frac {T_k}{N}
\right)^{\gamma-1}
\cdot
\frac {\tau_{k+1}}{N}
1 \left(
\frac {T_k}{N} < t 
\right)
\to
\frac {\lambda}{\pi}
\int_0^t 
\gamma s^{\gamma-1} ds
\eeq
as 
$N \to \infty$. \\
(2)
We first suppose 
$I = [t_1, t_2]$.
Since 
\beq
&&
{\bf P}
\left(
\mu_{\lambda}^{(n)}[t_1, t_2] = 0
\right)
\\
& \le &
{\bf E}
\left[
\prod_{k \ge 0}
{\bf P}
\left[
N_k^- = 0 
\middle|
{\cal C}_k, (\tau_i)_i
\right]
1 \left(
\frac {T_{k+1}}{N} \ge t_1, 
\;
\frac {T_k}{N} \le t_2
\right)
\right]
+
{\bf P}({\cal C}^c)
\eeq
and since 
\beq
{\bf P}
\left[
N_k^- = 0
\middle|
{\cal C}_k
\right]
\to
{\bf E}
\left[
\exp 
\left(
- \frac {\tau_{k+1}}{N}
\cdot
\frac {\lambda}{\pi}
\cdot
\gamma
\left(
\frac {T_{k+1}}{N}
\right)^{\gamma-1}
\right)
\right]
\eeq
we have 
\beq
&&
\limsup
{\bf P}
\left(
\mu_{\lambda}^{(n)}[t_1, t_2] = 0
\right)
\\
& \le &
{\bf E}
\left[
\prod_{k \ge 0}
\exp 
\left(
- \frac{\lambda}{\pi}
\gamma
\left(
\frac {T_{k+1}}{N}
\right)^{\gamma-1}
\cdot
\frac {\tau_{k+1}}{N}
\right)
1 \left(
\frac {T_{k+1}}{N} \ge t_1, 
\; 
\frac {T_k}{N} \le t_2
\right)
\right].
\eeq
Taking 
$N \to \infty$ 
proves 
(2) 
for
$I = [t_1, t_2]$. 
General case 
easily follows from the Markov property. 
\QED\\

\noindent
{\it Proof of Lemma \ref{order}}\\
We decompose 
${\bf P}({\cal E}_u)$ 
as follows. 
\begin{eqnarray}
&&
{\bf P}({\cal E}_u)
\nonumber
\\
& \le &
{\bf P}
\left(
{\cal E}_u 
\cap
\{
\zeta_u \in [u_0, u]
\}
\right)
+
{\bf P}
\left(
{\cal E}_u \cap 
\{ \zeta_u < u_0 \}
\right)
\nonumber
\\
& \le &
{\bf P}
\left(
{\cal E}_u 
\cap
\{
\zeta_u \in [u_0, u]
\}
\right)
+
{\bf P}
\left(
\bigcap_{s \in [u_0, u]}
{\cal E}_s
\right)
\nonumber
\\
& \le &
{\bf P}
\left(
\{
\zeta_u \in [u_0, u]
\}
\right)
+
{\bf P}
\left(
\bigcap_{s \in [u_0, u]}
{\cal E}_s
\cap
\Biggl\{
\{ 
\Theta^{(n)}_{\lambda}(u_0)
\}
\le
2 \arctan
n^{- \frac {1}{4\gamma}}
\Biggr\}
\right)
\nonumber
\\
&& \qquad +
{\bf P}
\left(
\Biggl\{
\{ \Theta^{(n)}_{\lambda}(u_0) \}
\ge
2 \arctan
n^{- \frac {1}{4\gamma}}
\Biggr\}
\right)
\nonumber
\\
& \le &
{\bf P}
\left(
\{
\zeta_u \in [u_0, u]
\}
\right)
+
n^{- \frac {c'}{\gamma}}
+
{\bf P}
\left(
\Biggl\{
\{ \Theta^{(n)}_{\lambda}(u_0) \}
\ge
2 \arctan
n^{- \frac {1}{4\gamma}}
\Biggr\}
\right)
\label{3terms}
\end{eqnarray}
where we used the monotonicity of 
$\lfloor 
\Theta^{(n)}_{\lambda, \lambda'} 
\rfloor$ 
in the 2nd inequality.
In the last inequality, 
we used the fact that, when  
$\{ \Theta^{(n)}_{\lambda}(u_0) \}_{\pi}
\le
2 \arctan
n^{- \frac {1}{4\gamma}}$, 
we necessarily have 
$\{ \Theta^{(n)}_{\lambda, \lambda'}(u_0) \}_{\pi}
\ge 
\pi - 
2 \arctan n^{- \frac {1}{4\gamma}}$.
Hence by Lemma \ref{explosion3} 
we have 
\beq
{\bf P}
\left(
\bigcap_{s \in [u_0, u]}
{\cal E}_s
\cap
\Biggl\{
\{ \Theta^{(n)}_{\lambda}(u_0) \}
\le
2 \arctan
n^{- \frac {1}{4\gamma}}
\Biggr\}
\right)
\le
n^{- \frac {c'}{\gamma}}, 
\eeq
proving the last inequality in 
(\ref{3terms}).
Now 
we integrate both sides of 
(\ref{3terms}) 
and use 
Lemma \ref{indicator}
for the 1st and 3rd terms of RHS. 
\QED
\\
%

For the proof of 
Lemmas \ref{monotonicity2}, \ref{independent}, 
let
$(\xi_i^{\lambda})$, 
$(\xi_i^{\lambda'})$, 
$(\xi_i^{\lambda', \lambda''})$
be the atoms of 
$P_{\lambda}$, 
$P_{\lambda'}$, 
$P_{\lambda', \lambda''}$
respectively.
Also, let 
$( \zeta^{\lambda}_i )$, 
$( \zeta^{\lambda'}_i )$, 
$( \zeta^{\lambda', \lambda''}_i )$
be the atoms of 
$\mu_{\lambda}^{(n)}$, 
$\mu_{\lambda'}^{(n)}$, 
$\mu_{\lambda', \lambda''}^{(n)}$
respectively.\\
%
%

%
\noindent
{\it Proof of Lemma \ref{monotonicity2}\\
For  
$N \in {\bf N}$, 
let 
\beq
p_N^n
&:=&
{\bf P}
\left(
\exists i \; : \; 
\zeta_i^{\lambda} < t, 
\;
\forall j \ge i, 
\;
| 
\zeta_i^{\lambda} - \zeta_j^{\lambda'}
|
>
\frac {1}{2N}
\right).
\eeq
It is then 
sufficient to show 
$\limsup_{n \to \infty} p_N^n = 0$.
Let 
$(T_k)_k$
be the random division of intervals used in the proof of Proposition \ref{marginal}.
Then we have 
\beq
p_N^n
& \le &
{\bf P}
\left(
\exists k \le [2Nt] +1
\, : \, 
\left\lfloor
\frac {\Theta^{(n)}_{\lambda}}{\pi}
\right\rfloor
\mbox{ jumps on }
\left[
\frac {T_k}{N}, \frac {T_{k}+2}{N}
\right]
\mbox{ but not }
\left\lfloor
\frac {\Theta^{(n)}_{\lambda'}}{\pi}
\right\rfloor
\right)
\\
& \le &
\sum_{k=1}^{[2Nt]+1}
{\bf P}
\left(
\left\{ 
\Theta^{(n)}_{\lambda'}
\left(
\frac {T_k}{N}
\right) 
\right\}_{\pi}
\le
\left\{ 
\Theta^{(n)}_{\lambda}
\left(
\frac {T_k}{N}
\right) 
\right\}_{\pi}
\right)
\eeq
where we used 
the monotonicity of 
$\lfloor\Theta^{(n)}_{\lambda, \lambda'}/\pi \rfloor$.
It thus suffices to use 
Lemma \ref{order}.
\QED\\

%
\noindent
{\it Proof of Lemma \ref{independent}}\\
As 
in the proof of Lemma \ref{monotonicity2}, 
it is sufficient to show 
\beq
p_N^n
&=&
{\bf P}
\left(
\exists i, j \in {\bf N}
\; : \;
\zeta_i^{\lambda} < t, 
\;
\zeta_j^{\lambda, \lambda'} < t, 
\;
| \zeta_i^{\lambda} - \zeta_j^{\lambda, \lambda'} |
<
\frac {1}{2N}
\right)
\eeq
satisfies 
$\limsup_{N \to \infty}
\limsup_{n \to \infty}
p_N^n = 0$.
\begin{eqnarray}
&&
p_N^n 
\nonumber
\\
& \le &
{\bf P}
\left(
\exists i, j \in {\bf N}, 
\;
\exists k \le [2Nt]+ 1, 
\;
\frac {T_k}{N} 
\le
\zeta_i^{\lambda}, \zeta_j^{\lambda, \lambda'}
\le
\frac {T_{k}+2}{N}
\right)
\nonumber
\\
&=&
{\bf P}
\left(
\exists k \le [2Nt] + 1 
\; : \; 
\left\lfloor
\frac {\Theta^{(n)}_{\lambda}}{\pi}
\right\rfloor, 
\left\lfloor
\frac {\Theta^{(n)}_{\lambda'} - \Theta^{(n)}_{\lambda}}{\pi}
\right\rfloor
\mbox{ both jump on }
\left[
\frac {T_k}{N}, \frac {T_{k}+2}{N}
\right]
\right)
\nonumber
\\
& \le &
\sum_{k=1}^{[2Nt]+1}
{\bf P}
\left(
\left\{ 
\Theta^{(n)}_{\lambda'}
\left(
\frac {T_k}{N}
\right) 
\right\}_{\pi}
\le
\left\{ 
\Theta^{(n)}_{\lambda}
\left(
\frac {T_k}{N}
\right) 
\right\}_{\pi}
\right)
\nonumber
\\
&&\qquad+
\sum
{\bf P}
\left(
\left\lfloor
\frac {\Theta^{(n)}_{\lambda'}}{\pi}
\right\rfloor
\mbox{ jumps more than 2-times on }
\left[
\frac {T_k}{N}, \frac {T_{k}+2}{N}
\right]
\right)
\label{2times}
\end{eqnarray}
where we used the monotonicity of 
$\lfloor\Theta^{(n)}_{\lambda, \lambda'}/\pi \rfloor$
in the last inequality. 
The 1st term 
in RHS of 
(\ref{2times}) 
has been estimated in the proof of Lemma \ref{monotonicity2}.
For the 2nd term, 
we use 
Proposition \ref{marginal}. 
\beq
&&
\limsup_{n \to \infty}
\sum_{k=1}^{[2Nt]+1}
{\bf P}
\left(
\mu_{\lambda'}^n
\left[
\frac {T_k}{N}, \frac {T_{k}+2}{N}
\right]
\ge 2
\right)
\\
& \le &
C
\sum_{k=1}^{[2Nt]+1}
{\bf E}
\Biggl[
\exp
\left[
- \frac {\lambda}{\pi}
\int_{\frac {T_k}{N}}^{\frac {T_{k}+2}{N}}
\gamma u^{\gamma-1} du
\right]
\cdot
\left(
- \frac {\lambda}{\pi}
\int_{\frac {T_k}{N}}^{\frac {T_{k}+2}{N}}
\gamma u^{\gamma-1} du
\right)^2
\Biggr]
\\
&=&
O(N^{-1}). 
\eeq
\QED
%

\vspace*{1em}
\noindent {\bf Acknowledgement }
One of the authors(F.N.) 
would like to thank Benedek Valk\'o 
for valuable discussions and letting him know the 
reference 
\cite{AD}. 
The authors 
would like to thank the referee for valuable suggestions to improve the readability of this paper.
The authors 
would also like to thank the Isaac Newton Institute for
Mathematical Sciences for its hospitality during the programme
``Periodic and Ergodic Spectral Problems"
supported by EPSRC Grant Number EP/K032208/1.
This work is partially supported by 
JSPS KAKENHI Grant 
Number 26400128(S.K.)
and
Number 26400145(F.N.).

%


\small
\begin{thebibliography}{99}

\bibitem{AD}
Allez, R., Dumaz, L., : 
From sine kernel to Poisson statistics, 
Elec. J. Prob. {\bf 19}(2014), 1-25.

\bibitem{AD2}
Allez, R., Dumaz, L., : 
Tracy-Widom at high temperature, 
J. Stat. Phys., 
{\bf 156}(2014), 1146-1183
%
\bibitem{ALS}
Avila, A.,  Last, Y., and Simon, B., :
Bulk Universality and Clock Spacing of zeros for Ergodic Jacobi Matrices with A.C. spectrum,
Anal. PDE {\bf 3}(2010), 81-108.

\bibitem{KiN}
Killip, R., Nakano, F., : 
Eigenfunction statistics in the localized Anderson model,
Annales Henri Poincar\'e. {\bf 8}, no.1 (2007), 27-36.
%
\bibitem{KS}
Killip, R., Stoiciu, M., : 
Eigenvalue statistics for CMV matrices : 
from Poisson to clock via random matrix ensembles, 
Duke Math. {\bf 146}(2009), 361-399.
%
\bibitem{KLS}
Kiselev, A., Last, Y., and Simon, B., : 
Modified Pr\"ufer and EFGP Transforms 
and the Spectral Analysis of One-Dimensional Schr\"odinger Operators, 
Commun. Math. Phys. {\bf 194}(1997), 1-45. 
%
\bibitem{KN}
Kotani, S., Nakano, F., : 
Level statistics for the one-dimensional Schr\"odinger operator with random decaying potentials, 
Interdisciplinary Mathematical Sciences, 
Vol. 17(2014), 343-373.
%
\bibitem{KU}
Kotani, S. Ushiroya, N. :
One-dimensional Schr\"odinger operators with random decaying 
potentials, 
Commun. Math. Phys. {\bf 115}(1988), 247-266. 
%
%
\bibitem{KVV}
Kritchevski, E., Valk\'o, B., Vir\'ag, B., : 
The scaling limit of the critical one-dimensional random 
Schr\"odinger operators, Commun. Math. Phys. {\bf 314}(2012), 775-806.
%
%
\bibitem{N1}
Nakano, F., :   
Distribution of localization centers in some discrete random systems, 
Rev. Math. Phys. {\bf 19}(2007), 941-965. 
%
\bibitem{N2}
Nakano, F., : 
Level statistics for one-dimensional Schr\"odinger operators and Gaussian beta ensemble, 
J. Stat. Phys. {\bf 156}(2014), 66-93.
%
\bibitem{N3}
Nakano, F., : 
Limit of Sine$_{\beta}$ and Sch$_{\tau}$ processes, 
RIMS Kokyuroku, {\bf 1970}(2015), 83 - 89.
%
\bibitem{N4}
Nakano, F., : 
Fluctuation of density of states for 1d Schr\"odinger operators, 
J. Stat. Phys. {\bf 166}(2017), 1393-1404.
%
\bibitem{VV}
Valk\'o, B. and Vir\'ag, V. : 
Continuum limits of random matrices and the Brownian carousel, 
Invent. Math. {\bf 177}(2009), 463-508.
%
\end{thebibliography}
\end{document}